\documentclass[11pt]{article}

\usepackage{graphicx}
\usepackage{url}
\usepackage{amsfonts}
\usepackage{amsmath}
\usepackage{amsthm}
\usepackage{amssymb}
\usepackage{url}
\usepackage{subfig}
\usepackage{hyperref}
\usepackage{fullpage}
\usepackage{algo}
\usepackage{framed}
\usepackage{multirow}

\allowdisplaybreaks

\newtheorem{theorem}{Theorem}%[section]

\newtheorem{lemma}[theorem]{Lemma}

\newtheorem{definition}[theorem]{Definition}

\usepackage[usenames,dvipsnames]{xcolor}

\DeclareMathOperator{\mbr}{\mathsf{MBR}} %
\DeclareMathOperator{\opt}{\mathsf{opt}} %
\DeclareMathOperator{\pess}{\mathsf{pess}} %

\DeclareMathOperator{\NE}{\mathtt{NE}} %
\DeclareMathOperator{\NW}{\mathtt{NW}} %
\DeclareMathOperator{\SE}{\mathtt{SE}} %
\DeclareMathOperator{\SW}{\mathtt{SW}} %

\title{Linear Separability in Spatial Databases}

\author{
Gilberto Guti\'errez\thanks{Computer Science and Information Technologies Department, Universidad del B\'io-B\'io, Chill\'an, Chile.
{\tt ggutierr@ubiobio.cl; ctorresf@alumnos.ubiobio.cl}.} 
\and
Pablo P\'erez-Lantero\thanks{Departamento de Matem\'atica y Ciencia de la Computaci\'on, Universidad de Santiago de Chile (USACH), Santiago, Chile.
{\tt pablo.perez@usach.cl}.}
\and
Claudio Torres\footnotemark[1]
}

\begin{document}

\maketitle

\begin{abstract}
Given two point sets $R$ and $B$ in the plane, with cardinalities $m$ and $n$, respectively, and each set stored in a separate R-tree, we present an algorithm to decide whether $R$ and $B$ are linearly separable. Our algorithm exploits the structure of the R-trees, loading into the main memory only relevant data, and runs in $O(m\log m + n\log n)$ time in the worst case. As experimental results, we implement the proposed algorithm and executed it on several real and synthetic point sets, showing that the percentage of nodes of the R-trees that are accessed and the memory usage are low in these cases. We also present an algorithm to compute the convex hull of $n$ planar points given in an R-tree, running in $O(n\log n)$ time in the worst case.
\end{abstract}

\section{Introduction}

Let $R$ be a finite set of red points and $B$ a finite set of blue points in the plane,
with cardinalities $m$ and $n$, respectively.
We say that $R$ and $B$ have {\em linear separability}, or that are {\em linearly separable},
if there exists a line such that: the elements of $R$ belong to one
of the halfplanes bounded by the line, the elements of $B$ belong
to the other halfplane, and if the line contains points from $R\cup B$, then
it contains points from exactly one between $R$ and $B$.
It is well known that $R$ and $B$ are linearly separable if and only 
if the convex hulls $conv(R)$ and $conv(B)$ have an empty intersection,
where $conv(X)$ denotes the convex hull of $X\subset\mathbb{R}^2$.
Because of this, deciding whether $R$ and $B$ are linearly separable, and in the positive 
case finding such a separating line, can be done within the following steps: 
compute $conv(R)$ and $conv(B)$ in times $O(m\log m)$ and $O(n\log n)$, respectively;
and test whether the intersection of $conv(R)$ and $conv(B)$ is empty in time $O(m+n)$~\cite{Toussaint85}.
In fact, if $conv(R)$ and $conv(B)$ has an empty intersection, then there exists
a separating line containing an edge of $conv(R)$ or $conv(B)$~\cite{Czyzowicz1992}.
Another way of deciding linear separability of $R$ and $B$, without computing
the convex hulls, is to formulate the separation
problem as a linear program (LP) in two variables and $m+n$ constraints, 
and use the algorithm of Megiddo~\cite{Megiddo84}
or the algorithm of Dyer~\cite{Dyer84} to solve the LP
in linear $O(m+n)$ time in the worst case. Other more practical option,
due to the big hidden constants in the $O(m+n)$ running times of these two algorithms,
is to use the simpler randomized algorithm of Seidel~\cite{Seidel91} that solves the LP
in expected $O(m+n)$ time.

In this paper, we study the problem of deciding whether $R$ and $B$ are linearly separable, and at 
the same time returning a separating line when the answer is positive, in the context of
the spatial databases: We are given as input an R-tree with the points of $R$ and a second R-tree with the 
points of $B$.
The R-tree is a secondary-storage, tree-like height-balanced data structure designed for the dynamic indexation of a set of dimensional geometric objects~\cite{Guttman1984,Manolopoulos2005}, and it is considered an standard in the context
of the spatial databases. See Section~\ref{sec:R-tree} for further details.
Because of the high data volume of the R-trees, loading from the R-trees all points of 
$R\cup B$ and running a known algorithm for testing linear
separability of $R$ and $B$ is considered impractical in this context. Then,
we aim to design an efficient algorithm working directly with the R-tree structure
and able of loading from the R-trees only the relevant data. Most of the algorithms in this context,
apart of producing the correct answer, aim to minimize both the running time and the number of 
nodes read from the R-trees, since each node is implemented as a disk page.

The spatial databases (SDB) represent an important aid for geographical information systems (GIS) 
to manage large amounts of data. However, SDB require the design of new data structures, 
spatial access methods, query languages, and algorithms to manage spatial information. 
In this sense, several algorithm have been designed for spatial queries such as the window query, 
the intersection query, the nearest neighbor, and the spatial join~\cite{Gaede1998,Shekhar2003}. 
Many of these queries are problems that were first tackled in the field of the computational
geometry, where it is assumed that all spatial objects fit into the main memory, and later, 
these problems were faced in the field of the SDB. Following this path, several algorithms have 
been proposed considering that objects are stored in a multidimensional structure, 
in most cases an R-tree~\cite{Guttman1984}.
For example, Corral et al.~\cite{Corral2004} and Hjaltason et al.~\cite{Hjaltason1998} 
presented several algorithms that solve the $k$-pairs ($k\ge 1$) of nearest neighbors 
between two sets, Roussopoulos et al.~\cite{Roussopoulos1995} showed an algorithm to find 
the nearest neighbor to a given point, 
Guti\'errez et al.~\cite{GutierrezPBC14} showed how to find a largest rectangle containing 
a query object and no point stored at an input R-tree, 
and B\"{o}hm and Kriegel~\cite{Bohm2001} described methods for computing the convex hull of point sets stored
in hierarchical index structures such as the R-tree or its variants. Among the geometric problems
in spatial databases, this later work is well related to the results of this paper since linear
the separability of two point sets can be decided by computing the convex hull of each set, and querying
the disjointness of the convex hulls. It is
worth noting that B\"{o}hm and Kriegel's algorithms do not exploit 
particular properties of R-trees such as the fact that the node regions are minimum bounding rectangles.
We explicitly exploit such a property in the algorithm that we propose for deciding the linear
separability of two point sets in the plane. Furthermore, in many cases our algorithm can decide the 
linear separability without computing such convex hulls.

The linear separability of two point sets, in either the plane or higher dimensions 
where a hyperplane separates the two point sets, is a concept well used in data mining
and classification problems~\cite{Bennett2000,Hand2001}. 
In this setting, where each point set represents the data of one class, 
% and under two-class datasets represented as linearly separable sets in an Euclidean
% space, 
the support vector machines (SVM) are a very robust
methodology for inference~\cite{Bennett2000}. The SVM is the hyperplane that separates 
the point sets and maximizes the minimum distance from the points to it.

%\Todo{Add motivation to consider linear separability in the plane in the GIS/SDB context}

The results of this paper are the following ones:
\begin{itemize}
	\item[(1)] We present an algorithm working directly with the R-trees of $R$ and $B$, able
    of deciding whether $R$ and $B$ are linearly separable, and in the positive case
    returning a separating line. In each step, it loads in the main memory data
    of only one level for each of the R-trees, and before descending in the R-trees to the
    next levels, to finally end at the leaf nodes, it verifies whether the gathered information is enough to
    decide the separability condition of $R$ and $B$. The asymptotic running time in the 
    worst case is $O(m\log m + n\log n)$. 
    % Then, from the theoretical point of view,
    % our algorithm in the worst case is equivalent to computing both convex hulls and verifying
    % whether they are disjoint. 
    
    \item[(2)] The techniques used in the separability testing algorithm can be extended to design
    an algorithm that computes the convex hull of a finite planar point set given as input
    in an R-tree. If $n$ denotes the number of input points, the asymptotic running time is
    $O(n\log n)$.
    
    \item[(3)] We implement the separability testing algorithm and executed it on several real and synthetic colored point
    sets, showing that in both cases the number of nodes of the R-trees that are accessed by the algorithm, and
    the amount of memory used, are low for these point sets. To generate synthetic point sets, we consider parameters
    such as the number of points to generate, the distribution of the points (uniform or Gaussian), and other parameters
    to define the positions of the minimum bounding rectangles of the red and blue points, respectively.
\end{itemize}

\paragraph{Notation:} Given a set $X\subset\mathbb{R}^2$, let $\mbr(X)$ denote the Minimum Bounding Rectangle 
(MBR) of $X$, which is the minimum-area rectangle that contains $X$. 
Every rectangle is considered axis-aligned in this paper.
Note that each of the four sides of $\mbr(X)$ contains at least one point of $X$.
Extending the notation,
if $Y$ is a set of subsets of
$\mathbb{R}^2$ (e.g.\ a set of rectangles), then $conv(Y)$ (resp.\ $\mbr(Y)$) denotes the convex
hull (resp.\ MBR) of all points contained in some element of $Y$, 
and $conv(Y\cup X)$ denotes
the convex hull of the union of the points of $conv(Y)$ and $conv(X)$.

\paragraph{Outline:} We continue this section by describing the R-tree data structure in Section~\ref{sec:R-tree},
and the idea of our separability algorithm in Section~\ref{sec:idea}. In Section~\ref{sec:preliminaries},
we present the concepts and geometric properties that our algorithm uses. Later, in Section~\ref{sec:algorithms},
we present our linear separability test algorithm, together with the algorithm to compute the convex hull
of a point set given in an R-tree. In Section~\ref{sec:experiments}, we show the experimentation results.
Finally, in Section~\ref{sec:conclusions}, we state the conclusions and further research.

\subsection{The R-tree}\label{sec:R-tree}

An R-tree is a generalization of the $\mathrm{B}^+$-trees designed for the dynamic indexation of a set of $k$-dimensional geometric objects~\cite{Manolopoulos2005}. It is a hierarchical, height balanced multidimensional data structure, designed to be used in secondary storage. In inner levels the indexed objects are represented by the $k$-dimensional Minimum Bounding Rectangles (MBRs), which bound their children. In this paper, we focus on two dimensions, therefore each MBR is an axis-aligned rectangle, represented only by its bottom-left and top-right vertices. By using the MBRs instead of the exact geometrical representations of the objects, its representational complexity is reduced to two points where the most important features of the spatial object (position and extent) are maintained. The MBR is an approximation widely employed, and the R-trees belong to the category of data-driven access methods, since their structure adapts itself to the MBRs distribution in the space.

\begin{figure}[t]
	\centering
	\includegraphics[scale=1.1,page=13]{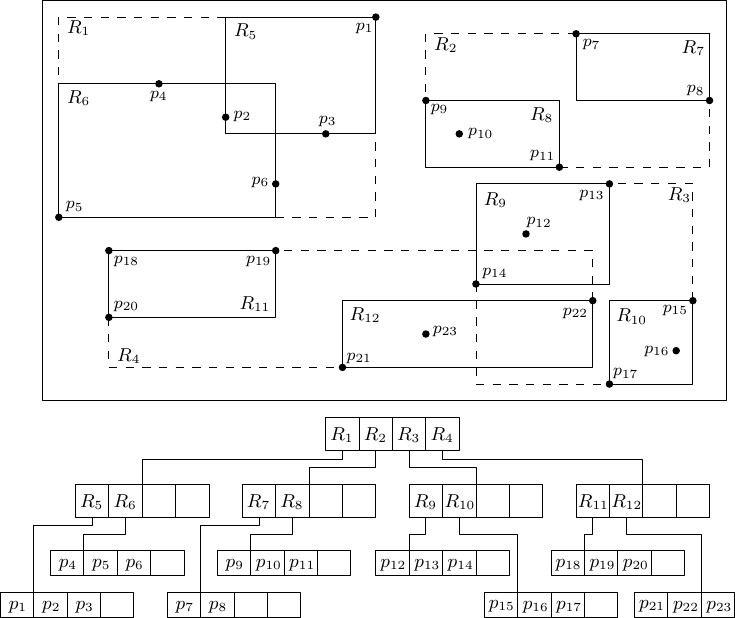}
	\caption{\small{
		An R-tree for the point set $\{p_1,p_2,\ldots,p_{23}\}$
        (picture based on one from~\cite{GutierrezPBC14}).
	}}
	\label{fig:R-tree}
\end{figure}

An R-tree for a finite point set $X\subset\mathbb{R}^2$ satisfies the following properties.
The leaves are on the same level,
and each leaf node contains indexed points of $X$. Every internal node contains entries of the form {\tt (MBR,ref)}, 
where {\tt ref} is a pointer to the child of the entry, and
{\tt MBR} is the minimum bounding rectangle of the {\tt MBR}'s 
(or the points if the child is a leaf node) contained in the entries of this child. 
An R-tree has the property that every node, except possibly the root, 
contains between $\mathsf{m}$ and $\mathsf{M}$ entries, 
where $2\le \mathsf{m} \le \lceil \mathsf{M}/2\rceil$. The root node contains at least two children nodes.
We will refer to the {\tt MBR} of an entry just as the rectangle of the entry, and
to the {\tt MBR}'s of the entries of a node just as the rectangles of the node.
For every entry (i.e.\ rectangle) of a node, the rectangles of the child node of the entry (i.e.
the child rectangles) are not necessarily pairwise disjoint,
so they can overlap between them. Furthermore, points of $X$ could be covered by different rectangles 
of the R-tree, although each point of $X$ appears only once in the leaf nodes.
All nodes of an R-tree are implemented as disk pages. 
We consider that the leaf nodes are at level $\mathsf{h}$ and the root is at level $0$,
where $\mathsf{h}$ is the height of the R-tree.

Figure~\ref{fig:R-tree} depicts an R-tree. Dotted lines denote the rectangles of the entries at the root node. The rectangles with solid lines are the rectangles in the entries of nodes parent of the leaves. Finally, the points are the indexed objects in the leaves of the R-tree.

\subsection{Idea of our algorithm}\label{sec:idea}

The general idea of our algorithm is based in the following
observation: Suppose that have loaded in the main memory a set $N_R$ of rectangles of the R-tree of $R$ so that
they all cover $R$, and a set of rectangles $N_B$ of the R-tree of $B$ so that they all cover $B$. 
If the convex hull $conv(N_R)$ of the rectangles of $N_R$ does not intersect the convex hull
$conv(N_B)$ of the rectangles of $N_B$, then $conv(R)$ and $conv(B)$ are disjoint since
$conv(R)\subseteq conv(N_R)$ and $conv(B)\subseteq conv(N_B)$, and $R$ and $B$ are hence linearly separable.

A more concrete idea is the following: 
According to the relative positions of $\mbr(R)$ and $\mbr(B)$, we choose
from $\mbr(R)$ a set $V_R$ of at most three vertices,
and a similar set $V_B$ from $\mbr(B)$. 
The idea of choosing $V_R$ and $V_B$ is that
$conv(R)$ and $conv(B)$ are disjoint if and only if
$conv(R\cup V_R)$ and $conv(B\cup V_B)$ are disjoint.
We start with $N_R$ being the set of 
the rectangles stored in the root node of the R-tree
of $R$, and $N_B$ being the set of the rectangles stored at the root node of the R-tree of $B$. 
Then, we iterate as follows:
If $conv(N_R\cup V_R)$ and $conv(N_B\cup V_B)$ are disjoint, then
we report a `yes' answer and build a separating line. 
Otherwise, for each rectangle of $N_R$ we take the region 
of points that can be ensured to belong to $conv(R\cup V_R)$ and
form the set of regions $N'_R$. A similar set $N'_B$ is formed from $N_B$.
If $conv(N'_R\cup V_R)$ and $conv(N'_B\cup V_B)$ are not disjoint, then
we report a `no' answer.
Otherwise, we filter the set $N_R$ so that the new $N_R$ contains 
only the rectangles (or points) that are relevant to 
decide the linear separation of $R$ and $B$, and `refine' $conv(N_R \cup V_R)$
by replacing each rectangle in $N_R$ by 
its respective child rectangles (or points) in the R-tree. We do a similar 
procedure with $N_B$ and continue the iteration. If at some
point in the iteration both $N_R$ and $N_B$ consist of only points, then the answer is given
by the intersection condition of $conv(N_R\cup V_R)$ and $conv(N_B\cup V_B)$.

It is worth noting that we test the linear separability condition via
computing the convex hulls (or approximations of the convex hulls, e.g.\ $conv(N_R\cup V_R)$) of the two point sets.
We do not use any asymptotic-faster linear-time LP separability
testing algorithm~\cite{Dyer84,Megiddo84,Seidel91} since in this case
the process of filtering rectangles is more expensive in time: We can discard a rectangle
if we can ensure that it is contained in the convex hull, and it can be 
done in logarithmic time, as we will see later in the paper. Otherwise, if we 
do not compute the convex hull, as it happens if we use any of the linear-time LP separability
testing algorithms, then to decide whether a rectangle can be discarded we should consider
the relative position of the rectangle with respect to the other rectangles,
and this is much more expensive than determining whether the rectangle is inside 
a convex hull.

\section{Preliminaries}\label{sec:preliminaries}

Considering $\mbr(R)$ and $\mbr(B)$, we make the following
definitions:
\begin{itemize}
	\item We say that
	$\mbr(R)$ and $\mbr(B)$ have a {\em corner} intersection if each rectangle
	contains exactly one vertex of the other one (see Figure~\ref{fig:corner}),
	or one of the rectangles is contained in the other and they share exactly one vertex.
	
	\item We say that $\mbr(R)$ and $\mbr(B)$ have a {\em side} intersection
	if one of the rectangles contains exactly two vertices of the other one (see Figure~\ref{fig:side}),
	and is not contained within it.
	
	\item We say that $\mbr(R)$ and $\mbr(B)$ have a {\em containment} intersection
	if the interior of one of the rectangles contains from the other rectangle the 
	four vertices (see Figure~\ref{fig:containment}),
	or two adjacent vertices with the other two ones contained in the boundary.
	
	\item We say that $\mbr(R)$ and $\mbr(B)$ have a {\em piercing} intersection
	if their interiors are not disjoint, and 
	no rectangle contains in the interior a vertex of the other one (see Figure~\ref{fig:piercing}).
\end{itemize} 

\begin{figure}[t]
	\centering
	\subfloat[]{
		\includegraphics[scale=0.8,page=1]{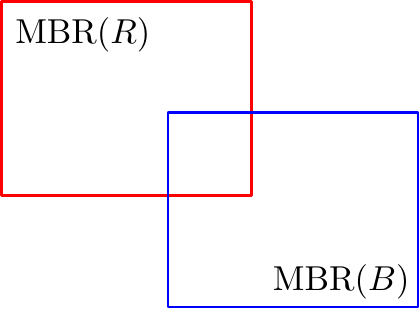}
		\label{fig:corner}
	}
	\hspace{0.3cm}
	\subfloat[]{
		\includegraphics[scale=0.8,page=2]{figs.pdf}
		\label{fig:side}
	}
	\hspace{0.3cm}
	\subfloat[]{
		\includegraphics[scale=0.8,page=3]{figs.pdf}
		\label{fig:containment}
	}
	\hspace{0.3cm}
	\subfloat[]{
		\includegraphics[scale=0.8,page=4]{figs.pdf}
		\label{fig:piercing}
	}
	\caption{\small{
		$\mbr(R)$ and $\mbr(B)$ have a:
		(a) corner intersection.
		(b) side intersection.
		(c) containment intersection.
		(d) piercing intersection.
	}}
	\label{fig:type-intersection}
\end{figure}

Up to symmetry, we assume without loss of generality throughout this paper
that the relative positions of $\mbr(R)$ and $\mbr(B)$, when they intersect,
are those shown in Figure~\ref{fig:type-intersection}.
Observe that if $\mbr(R)$ and $\mbr(B)$ have a piercing intersection, then
$R$ and $B$ are not linearly separable. This is because their convex hulls
have a non-empty intersection since each side of $\mbr(R)$ contains a red point
and each side of $\mbr(B)$ contains a blue point. Furthermore, the piercing 
intersection definition includes the case where $\mbr(R)=\mbr(B)$.
In the other cases of 
intersection (corner, side, and containment) the linear separation condition cannot be
directly deduced (see for examples Figure~\ref{fig:type-intersection2}).
In the trivial case where $\mbr(R)$ and $\mbr(B)$ do not intersect,
$R$ and $B$ are linearly separable with a vertical or horizontal line.

\begin{figure}[t]
	\centering
	\subfloat[]{
		\includegraphics[scale=1.0,page=5]{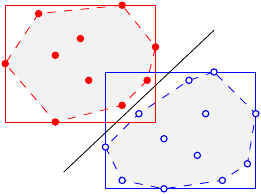}
		\label{fig:corner1}
	}
	\hspace{0.3cm}
	\subfloat[]{
		\includegraphics[scale=1.0,page=7]{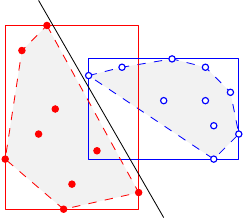}
		\label{fig:side1}
	}
	\hspace{0.3cm}
	\subfloat[]{
		\includegraphics[scale=1.0,page=9]{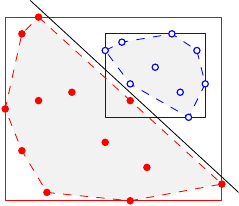}
		\label{fig:containment1}
	} \\
	\subfloat[]{
		\includegraphics[scale=1.0,page=6]{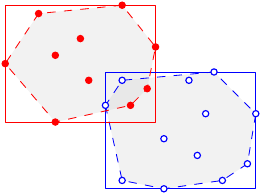}
		\label{fig:corner2}
	}
	\hspace{0.3cm}
	\subfloat[]{
		\includegraphics[scale=1.0,page=8]{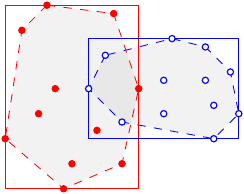}
		\label{fig:side2}
	}	
	\hspace{0.3cm}
	\subfloat[]{
		\includegraphics[scale=1.0,page=10]{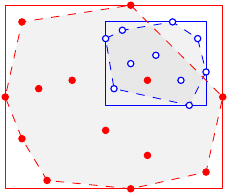}
		\label{fig:containment2}
	}
	\caption{\small{
		(a,b,c) Corner, side, and containment intersections of $\mbr(R)$ and $\mbr(B)$,
		where $R$ and $B$ have linear separability.
		(d,e,f) Corner, side, and containment intersections of $\mbr(R)$ and $\mbr(B)$,
		where $R$ and $B$ are not linearly separable. In each picture,
		the red points are represented as solid dots, and the blue points
		as tiny disks.
	}}
	\label{fig:type-intersection2}
\end{figure}

Our algorithm to decide whether $R$ and $B$ are linearly separable starts by detecting the 
type of intersection between $\mbr(R)$ and $\mbr(B)$. If they do not intersect, then we report
a `yes'. If they have a piercing intersection, then we report a `no'. Otherwise, if the
intersection is of type corner, side, or containment, then we need to elaborate a procedure that
gives the correct answer and at the same time returns a line separating $R$ and $B$ if the answer is
`yes'.
Suppose that $\mbr(R)$ and $\mbr(B)$ have a containment intersection, with $\mbr(B)$
contained in $\mbr(R)$. In this case, the rectangles $\mbr(R)$ and 
$\mbr(B\cup \{v\})$ have a corner intersection, where $v$ is any of the four vertices 
of $\mbr(R)$. Furthermore, if $conv(R)$ and $conv(B)$ are disjoint, then $conv(R)$ and $conv(B\cup \{v\})$
are also disjoint for some vertex $v$ of $\mbr(R)$. Conversely, if $conv(R)$ and $conv(B)$ are not disjoint,
then $conv(R)$ and $conv(B\cup \{v\})$ will not be disjoint for every vertex $v$ of $\mbr(R)$
since $conv(B)$ is contained in $conv(B\cup \{v\})$. Hence,
when $\mbr(R)$ and $\mbr(B)$ have a containment intersection, we can test whether
$R$ and $B$ are linearly separable by calling four times the test for the case where 
$\mbr(R)$ and $\mbr(B)$ have a corner intersection. That is, the answer is `yes' if and
only if $R$ and $B\cup\{v\}$ are linearly separable for at least one vertex $v$ of $R$.
Because of this reduction, we will consider in the following only intersections of
type corner or side.

\begin{definition}\label{def:VR-VB}
If $\mbr(R)$ and $\mbr(B)$ have a corner intersection, then let $V_R$ denote
the set of the top-right, top-left, and bottom-left vertices of $\mbr(R)$, and 
$V_B$ denote the set of the top-right, bottom-left, and bottom-right vertices of $\mbr(B)$.
If $\mbr(R)$ and $\mbr(B)$ have a side intersection, then let $V_R$ denote
the set of the top-left and bottom-left vertices of $\mbr(R)$, and 
$V_B$ denote the set of the top-right and bottom-right vertices of $\mbr(B)$ 
(see figures~\ref{fig:corner3} and~\ref{fig:side3}).
\end{definition}

\begin{figure}[t]
	\centering
	\subfloat[]{
		\includegraphics[scale=1.1,page=11]{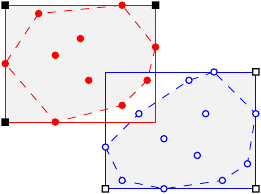}
		\label{fig:corner3}
	}
	\hspace{0.2cm}
	\subfloat[]{
		\includegraphics[scale=1.1,page=12]{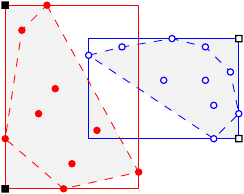}
		\label{fig:side3}
	}	
	\hspace{0.2cm}
	\subfloat[]{
		\includegraphics[scale=1.1,page=14]{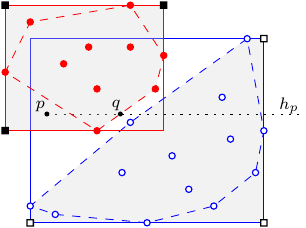}
		\label{fig:VB-VB-lemma}
	}
	\caption{\small{
		The sets of vertices $V_R$ and $V_B$ for: 
		(a) corner intersection; (b) side intersection.
		The vertices of $V_R$ are denoted as filled squares, and the vertices of
		$V_B$ as empty squares. 
		(c) Proof of Lemma~\ref{lem:VR-VB}.
		In each picture, the convex hulls
		$conv(R\cup V_R)$ and $conv(B\cup V_B)$ are denoted as shaded regions.
	}}
	\label{fig:VR-VB}
\end{figure}

The idea behind the definition of $V_R$ and $V_B$ is given in the next lemmas:

\begin{lemma}\label{lem:VR-VB}
The convex hulls $conv(R)$ and $conv(B)$ are disjoint if and only if the convex hulls
$conv(R\cup V_R)$ and $conv(B\cup V_B)$ are disjoint. 
\end{lemma}

\begin{proof}
If $conv(R\cup V_R)$ and $conv(B\cup V_B)$ are disjoint, then trivially 
$conv(R)$ and $conv(B)$ are also disjoint because $conv(R)\subseteq conv(R\cup V_R)$ and
$conv(B)\subseteq conv(B\cup V_B)$.
Then, suppose that $conv(R)$ and $conv(B)$ are disjoint (see Figure~\ref{fig:VB-VB-lemma}).
Let $p$ be a point of $conv(R\cup V_R)\setminus conv(R)$. 
If $p$ does not belong to $\mbr(B)$, then $p$ is not contained in
$conv(B\cup V_B)$. Then, assume that $p$ belongs to the intersection
$\mbr(R)\cap \mbr(B)$.
Let $h_p$ be the halfline with apex at $p$, and
oriented rightwards or downwards, such that $h_p$ contains a point $q$ of $conv(R)\cap \mbr(B)$. 
Note that 
$h_p$ always exists given $p$, and assume without loss of generality that $h_p$ is horizontal. 
The case where $h_p$ is vertical (which appears only in the case of a corner intersection) is analogous.
Since $q$ does not belong to $conv(B)$ because 
$conv(R)$ and $conv(B)$ are disjoint, and $q$ is in the interior of $\mbr(R)\cap \mbr(B)$,
$q$ is both to the left and above of $conv(B)$ in the case of a corner intersection, 
and to the left of $conv(B)$ in the
case of a side intersection.
Then, $q$ does not belong to $conv(B\cup V_B)\setminus conv(B)$ because of
the definition of $V_B$. 
Hence, $p$ is not in $conv(B\cup V_B)\setminus conv(B)$ because of the convexity of $conv(B)$
and that $p$ is to left of $q$. Similar symmetric
arguments show that if a point $p'$ belongs to $conv(B\cup V_B)\setminus conv(B)$, then 
$p'$ is not in $conv(R\cup V_R)\setminus conv(R)$. All of these observations
imply that $conv(R\cup V_R)$ and $conv(B\cup V_B)$ are disjoint.
\end{proof}

\begin{lemma}\label{lem:VR-VB-inter}
The convex hulls $conv(R\cup V_R)$ and $conv(B\cup V_B)$
have a non-empty intersection if and only if one of them contains a vertex of the other one.
Furthermore, the vertices of $conv(R\cup V_R)$ contained in $conv(B\cup V_B)$ are all consecutive,
and the vertices of $conv(B\cup V_B)$ contained in $conv(R\cup V_R)$ are all consecutive.
\end{lemma}

\begin{proof}
If one of $conv(R\cup V_R)$ and $conv(B\cup V_B)$ contains a vertex of the other one, 
then they are not disjoint. Suppose now that $conv(R\cup V_R)$ and $conv(B\cup V_B)$ 
are not disjoint. Let $X$ be a set of red points and $Y$ a set of blue points. 
If $conv(X)$ and $conv(Y)$ are not disjoint and neither of them contains a vertex
of the other one (see e.g.\ the piercing intersection of Figure~\ref{fig:piercing}),
then the convex hull $conv(X\cup Y)$ contains at least four bichromatic edges (i.e.\ edges connecting
points of different colors). For $X=R\cup V_R$ and $Y=B\cup V_B$, 
$conv(X\cup Y)$ contains only two bichromatic edges, given the relative positions
of $\mbr(R)$ and $\mbr(B)$ and the definitions of $V_R$ and $V_B$. Hence,
one of $conv(R\cup V_R)$ and $conv(B\cup V_B)$ must contain a vertex of the other one.
Furthermore, the fact that $conv(X\cup Y)$ contains only two bichromatic edges implies
the second part of the lemma.
\end{proof}

\section{Optimistic and pessimistic Convex Hulls}\label{sec:hulls}

The idea in this section is the following: Suppose that we have loaded a set of rectangles
from the R-tree of $R$, and a set of rectangles from the R-tree of $B$. We explain
a way of determining from these two sets of rectangles whether we have enough information 
to decide that $conv(R\cup V_R)$
and $conv(B \cup V_B)$ are disjoint or that they are not disjoint, without going deeper
in the R-tree loading more rectangles or points.

\begin{figure}[t]
	\centering
	\subfloat[]{
		\includegraphics[scale=1.1,page=15]{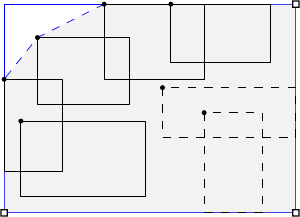}
		\label{fig:NB1}
	}
	\hspace{0.4cm}
	\subfloat[]{
		\includegraphics[scale=1.1,page=16]{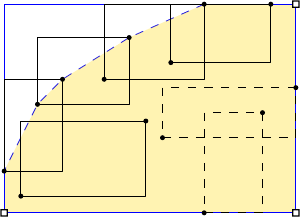}
		\label{fig:NB2}
	}	
	\caption{\small{
		A set $N_B$ of rectangles from the R-tree of $B$, in the case
		where $\mbr(R)$ and $\mbr(B)$ have a corner intersection.
		(a) The optimistic convex hull $\opt(N_B)$.
		(b) The pessimistic convex hull $\pess(N_B)$. The rectangles
		in dashed lines can be removed from $N_B$ since they
		are contained in $\pess(N_B)$.
	}}
	\label{fig:NB-corner}
\end{figure}

\begin{definition}\label{def:NR-NB}
Let $N_R$ be a set of rectangles from the R-tree of $R$, and 
$N_B$ a set of rectangles from the R-tree of $B$, such that
the following properties are satisfied:
\begin{itemize}
	\item[(1)] $conv(R\cup V_R)\subseteq conv(N_R\cup R)$.
	\item[(2)] $conv(B\cup V_B)\subseteq conv(N_B\cup B)$.
\end{itemize}
\end{definition}

For examples of Definition~\ref{def:NR-NB}, refer to Figure~\ref{fig:NB-corner}
and Figure~\ref{fig:NB-side}.
Since a point can be seen as a rectangle of null perimeter,
we extend the definitions of $N_R$ and $N_B$ so that these sets
can be made of points. 

\begin{definition}\label{def:optimistic}
Given the sets $N_R$ and $N_B$, the optimistic convex hull of $R\cup V_R$ is 
the set $\opt(N_R)=conv(N_R\cup V_R)$ which contains $conv(R\cup V_R)$, and the optimistic
convex hull of $B\cup V_B$ is set $\opt(N_B)=conv(N_B\cup V_B)$ 
which contains $conv(B\cup V_B)$ (see Figure~\ref{fig:NB1} and Figure~\ref{fig:NB3}).
\end{definition}

The idea of defining the optimistic convex hulls is the following observation:
If $\opt(N_R)$ and $\opt(N_B)$ are disjoint, then we can ensure that $conv(R\cup V_R)$ and $conv(B\cup V_B)$
are disjoint, and then $R$ and $B$ are linearly separable because of Lemma~\ref{lem:VR-VB}.
Furthermore, $\opt(N_R)$ and $\opt(N_B)$ are approximations to $conv(R\cup V_R)$ and $conv(B\cup V_B)$,
and to compute them we do not need the points covered by their rectangles, which are
located in the leaves of the R-trees.
%, and it is expensive to read them all. 

We also need a method for determining from $N_R$ and $N_B$ whether there is enough
information to decide that $conv(R\cup V_R)$ and $conv(B\cup V_B)$ are not disjoint.
This is explained in what follows.

Let $N$ be a rectangle: $\NE(N)$ denotes the north-east triangle of $N$, that is,
the subset of points of $N$ in or above the diagonal connecting the top-left and
bottom-right vertices.  Similarly, $\NW(N)$ denotes the subset of points of $N$ 
in or above the diagonal connecting the top-right and
bottom-left vertices; $\SE(N)$ denotes the subset of points of $N$ 
in or below the diagonal connecting the top-right and
bottom-left vertices; and $\SW(N)$ denotes the subset of points of $N$ 
in or below the diagonal connecting the top-left and
bottom-right vertices. 

Suppose that $\mbr(R)$ and $\mbr(B)$ have a corner intersection, and
let $N$ be a rectangle of $N_R$. Since in our model of R-tree every rectangle is
a minimum bounding rectangle and thus contains points of the represented point set
in every side, the set $\NW(N)$ is contained in the convex hull $conv(R\cup V_R)$. Similarly,
if $N$ is a rectangle of $N_B$, then the set $\SE(N)$ is contained in the convex hull $conv(B\cup V_B)$.
We use these observations to define the pessimistic convex hulls, which are always contained
in our goal $conv(R\cup V_R)$ and $conv(B\cup V_B)$.

\begin{definition}\label{def:pessmistic-corner}
Let $R$ and $B$ be red and blue point sets such that $\mbr(R)$ and $\mbr(B)$ have a corner intersection.
The pessimistic convex hull of $R\cup V_R$ is 
the set $\pess(N_R)=conv(N'_R\cup V_R)$ which is contained in $conv(R\cup V_R)$,
where $N'_R=\{\NW(N)\mid N\in N_R\}$.
The pessimistic convex hull of $B\cup V_B$ is 
the set $\pess(N_B)=conv(N'_B\cup V_B)$ which is contained in $conv(B\cup V_B)$,
where $N'_B=\{\SE(N)\mid N\in N_B\}$ (see Figure~\ref{fig:NB2}).
\end{definition}

\begin{figure}[t]
	\centering
	\subfloat[]{
		\includegraphics[scale=1.1,page=17]{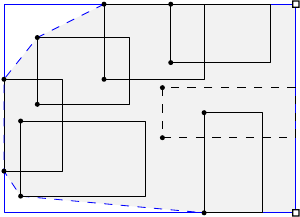}
		\label{fig:NB3}
	}
	\hspace{0.4cm}
	\subfloat[]{
		\includegraphics[scale=1.1,page=19]{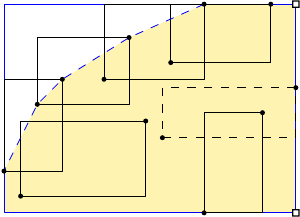}
		\label{fig:NB4}
	}
	\hspace{0.4cm}
	\subfloat[]{
		\includegraphics[scale=1.1,page=20]{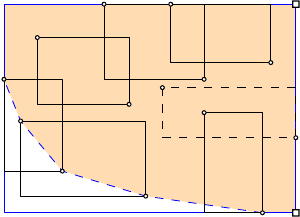}
		\label{fig:NB5}
	}
	\hspace{0.4cm}
	\subfloat[]{
		\includegraphics[scale=1.1,page=18]{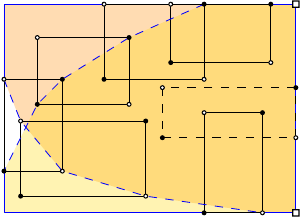}
		\label{fig:NB6}
	}	
	\caption{\small{
		A set $N_B$ of rectangles from the R-tree of $B$, when
		$\mbr(R)$ and $\mbr(B)$ have a side intersection.
		(a) The optimistic convex hull $\opt(N_B)$.
		(b) The convex hull $conv(N'_B\cup V_B \cup \{v_1\})$. 
		(c) The convex hull $conv(N''_B\cup V_B \cup \{v_2\})$.
		(d) The pessimistic convex hull $\pess(N_B)=conv(N'_B\cup V_B \cup \{v_1\})\cap conv(N''_B\cup V_B \cup \{v_2\})$.
		The rectangle in dashed lines can be removed from $N_B$ since it
		is contained in $\pess(N_B)$.
	}}
	\label{fig:NB-side}
\end{figure}

\begin{definition}\label{def:pessmistic-side}
Let $R$ and $B$ be red and blue point sets such that $\mbr(R)$ and $\mbr(B)$ have a side intersection.
The pessimistic convex hull of $R\cup V_R$ is 
the set $\pess(N_R)=conv(N'_R\cup V_R \cup \{u_1\}) \cap conv(N''_R\cup V_R \cup \{u_2\})$,
where
$N'_R=\{\NW(N)\mid N\in N_R\}$, $N''_R=\{\SW(N)\mid N\in N_R\}$, 
and $u_1$ and $u_2$ are the top-right and bottom-right vertices of $\mbr(R)$,
respectively.
The pessimistic convex hull of $B\cup V_B$ is 
the set $\pess(N_B)=conv(N'_B\cup V_B \cup \{v_1\})\cap conv(N''_B\cup V_B \cup \{v_2\})$,
where
$N'_B=\{\SE(N)\mid N\in N_B\}$, $N''_B=\{\NE(N)\mid N\in N_B\}$, 
and $v_1$ and $v_2$ are the bottom-left and top-left vertices of $\mbr(B)$,
respectively.
(see Figure~\ref{fig:NB-side}).
\end{definition}

When $\mbr(R)$ and $\mbr(B)$ have a corner intersection, the facts 
$\pess(N_R)\subseteq conv(R\cup V_R)$ and $\pess(N_B)\subseteq conv(B\cup V_B)$
are clear. In such a case of intersection, if $\pess(N_R)$ and $\pess(N_B)$ are not disjoint, then we can ensure that 
$conv(R\cup V_R)$ and $conv(B\cup V_B)$
are not disjoint, and then $R$ and $B$ are not linearly separable because of Lemma~\ref{lem:VR-VB}.
When $\mbr(R)$ and $\mbr(B)$ have a side intersection, the same
facts are proved in the following lemma, and then we can ensure that 
$conv(R\cup V_R)$ and $conv(B\cup V_B)$ are not disjoint if $\pess(N_R)$ and $\pess(N_B)$ are not.

\begin{lemma}\label{lem:pess-side}
Let $R$ and $B$ be red and blue point sets such that $\mbr(R)$ and $\mbr(B)$ have a side intersection.
The pessimistic convex hulls $\pess(N_R)$ and $\pess(N_B)$ are contained in 
$conv(R\cup V_R)$ and $conv(B\cup V_B)$, respectively.
\end{lemma}

\begin{figure}[t]
	\centering
	\subfloat[]{
		\includegraphics[scale=1.2,page=21]{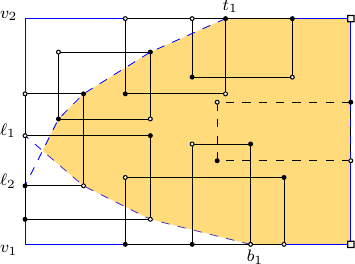}
		\label{fig:pess-side1}
	}
	\hspace{0.2cm}
	\subfloat[]{
		\includegraphics[scale=1.2,page=22]{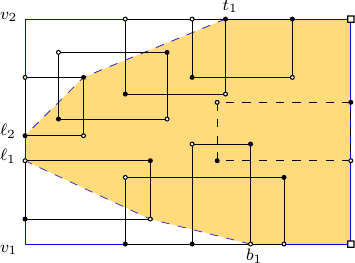}
		\label{fig:pess-side2}
	}	
	\caption{\small{
		Proof of Lemma~\ref{lem:pess-side}.
	}}
	\label{fig:pess-side}
\end{figure}

\begin{proof}
We will prove that $\pess(N_B)\subseteq conv(B\cup V_B)$. The arguments to prove
$\pess(N_R)\subseteq conv(R\cup V_R)$ are analogous. 
From the rectangles of
$N_B$ whose top sides are aligned with the top side of $\mbr(B)$, let $t_1$ a left-most of
the top-right vertices (see Figure~\ref{fig:pess-side}).
From the rectangles of
$N_B$ whose bottom sides are aligned with the bottom side of $\mbr(B)$, let $b_1$ a left-most of
the bottom-right vertices. 
From the rectangles of $N_B$ whose left sides are aligned with the left side
of $\mbr(B)$, let $\ell_1$ be a bottom-most of the top-left vertices, 
and $\ell_2$ a top-most of the bottom-left vertices.
Observe that $\ell_2$ and $t_1$ are vertices of $conv(N'_B\cup V_B\cup\{v_1\})$,
and that $b_1$ and $\ell_1$ are vertices of $conv(N''_B\cup V_B \cup\{v_2\})$.
Let $U$ denote the path
connecting $\ell_2$ with $t_1$ along the boundary of $conv(N'_B\cup V_B\cup\{v_1\})$
and clockwise, and $L$ denote the path connecting $b_1$ with $\ell_1$
along the boundary of $conv(N''_B\cup V_B\cup\{v_2\})$ and clockwise.
Note that there are the following blue points of $B$: a left-most blue point $t$ in the 
segment connecting $v_2$ and $t_1$, a left-most blue point $b$ in the segment connecting
$v_1$ and $b_1$, and both a top-most blue point $\ell$ and a bottom-most blue point $\ell'$ (possibly
equal to $\ell$) in the edge connecting $v_1$ and $v_2$. By the definitions
of $\ell_1$ and $\ell_2$, we have that $\ell$ belongs to the segment 
connecting $\ell_2$ and $v_2$, and $\ell'$ belongs to the segment 
connecting $\ell_1$ and $v_1$.
The key observation is that $b$, $\ell'$, $\ell$, and $t$ are all vertices of $conv(B\cup V_B)$.
Furthermore, the clockwise path along the boundary of $conv(B\cup V_B)$ that connects
$b$ and $\ell''$ is below the path $L$, and the similar path that connects
$\ell$ and $t$ is above the path $U$. 
All of these observations imply that $\pess(N_B)\subseteq conv(B\cup V_B)$.
\end{proof}

\section{Algorithms}\label{sec:algorithms}

In this section, we present our separability testing algorithm for $R$ and $B$ in the
R-tree model. 
We start by presenting the ingredient algorithms for the separability algorithm:
computation of the optimistic and pessimistic convex hulls (Section~\ref{sec:hulls-computation}),
deciding whether the convex hulls are disjoint and finding a separating
line in the positive case (Section~\ref{sec:hull-inter}),
and filtering the sets of rectangles $N_R$ and $N_B$ (Section~\ref{sec:filtering}).
Then, we show the separability algorithm (Section~\ref{sec:separab-algo}).
Finally, we show how the techniques to previous computations can be used to
compute the convex hull of a point set given in an R-tree (Section~\ref{sec:mono-hull-compute}).

\subsection{Convex hulls computation}\label{sec:hulls-computation}

We explain how to compute the optimistic and pessimistic
convex hulls for $N_B$ in both cases
of intersections of $\mbr(R)$ and $\mbr(B)$: corner and side. 
The algorithms to compute these convex hulls for $N_R$
are analogous by symmetry.

Let $t_1$ (resp.\ $t_2$) be a left-most vertex from the top-right (resp.\ top-left) 
vertices of the rectangles of $N_B$ whose top sides are aligned with the top side of $\mbr(B)$;
$b_1$ (resp.\ $b_2$) a left-most vertex from the bottom-right (resp.\ bottom-left) 
vertices of the rectangles of $N_B$ whose bottom sides are aligned with the bottom side of $\mbr(B)$;
$\ell_1$ (resp.\ $\ell'_1$) a bottom-most vertex from the top-left (resp.\ bottom-left) 
vertices of the rectangles of $N_B$ whose left sides are aligned with the left side of $\mbr(B)$; and
$\ell_2$ (resp.\ $\ell'_2$) a top-most vertex from the bottom-left (resp.\ top-left) 
vertices of the rectangles of $N_B$ whose left sides are aligned with the left side of $\mbr(B)$.
All of these points can be found in $O(|N_B|)$ time, by a single pass over the elements of $N_B$ 
(see Figure~\ref{fig:alg}).

Suppose that $\mbr(R)$ and $\mbr(B)$ have a corner intersection.
To compute $\opt(N_B)$, we need to compute the convex hull
of the top-left vertices of the rectangles of $N_B$ and the set $V_B$.
Observe that such vertices that are not in the triangle $T$ with vertex 
set $\{v_2,\ell'_2,t_2\}$ cannot be vertices of $\opt(N_B)$ (see Figure~\ref{fig:alg1}).
Then, we first find in $O(|N_B|)$ time
the set $S$ of the top-left vertices of the rectangles of $N_B$ 
which belong to $T$,
and after that compute $\opt(N_B)=conv(S\cup V_B)$ in $O(|S|\log |S|)$ time,
using a standard algorithm for computing the convex hull.
Doing this, we apply the convex hull algorithm for only the relevant set of points.
Ideas similar to these ones are going to be used in the following.
To compute $\pess(N_B)$, we find in $O(|N_B|)$ time the set $S'$
of the bottom-left and top-right vertices of the rectangles of $N_B$ 
that belong to the triangle with vertices $\{v_2,\ell_2,t_1\}$ (see Figure~\ref{fig:alg1}),
and then compute $\pess(N_B)=conv(S'\cup V_B)$ in $O(|S'|\log |S'|)$ time.

\begin{figure}[t]
	\centering
	\subfloat[]{
		\includegraphics[scale=1.1,page=27]{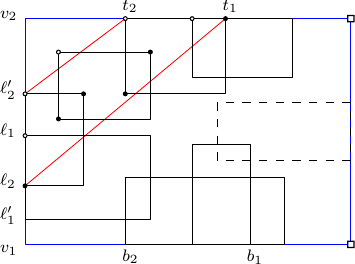}
		\label{fig:alg1}
	}
	\hspace{0.2cm}
	\subfloat[]{
		\includegraphics[scale=1.1,page=28]{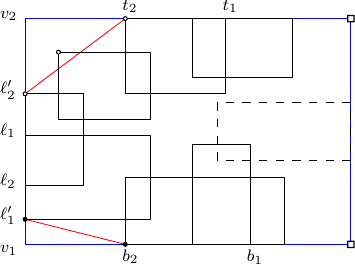}
		\label{fig:alg2}
	}		
	\hspace{0.2cm}
	\subfloat[]{
		\includegraphics[scale=1.1,page=29]{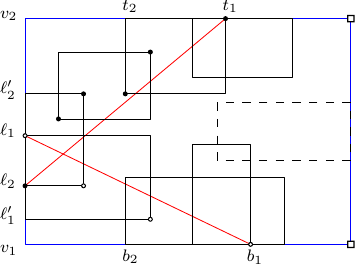}
		\label{fig:alg3}
	}	
	\caption{\small{
		Algorithm to compute $\opt(N_B)$ and $\pess(N_B)$:
		(a) The filled vertices are the points of $S$, and the hollow vertices are the points
		of $S'$.
		(b) The vertices of the set $S_0$.
		(c) The filled vertices are the points of $S_1$, and the hollow vertices are the points
		of $S_2$.
	}}
	\label{fig:alg}
\end{figure}

When $\mbr(R)$ and $\mbr(B)$ have a side intersection, we proceed as follows:
To compute $\opt(N_B)$, we compute in $O(|N_B|)$ time the point set $S_0$
consisting of the top-left vertices of the rectangles of $N_B$ that belong to the triangle with 
vertices $\{v_2,\ell'_2,t_2\}$, and the bottom-left vertices that belong to the triangle with 
vertices $\{v_1,\ell'_1,b_2\}$ (see Figure~\ref{fig:alg2}). 
Then, we compute $\opt(N_B)=conv(S_0\cup V_B)$ in $O(|S_0|\log |S_0|)$ time.
To compute $\pess(N_B)$, we first find in $O(|N_B|)$ time the set
$S_1$ containing the bottom-left and top-right vertices of the rectangles in $N_B$
that belong to the triangle with vertices $\{v_2,\ell_2,t_1\}$, and the
set $S_2$ containing the top-left and bottom-right vertices of the rectangles in $N_B$
that belong to the triangle with vertices $\{v_1,\ell_1,b_1\}$ (see Figure~\ref{fig:alg3}).
After that, we compute the convex hulls
$C_1=conv(N'_B\cup V_B\cup \{v_1\})=conv(S_1\cup V_B\cup \{v_1\})$ and 
$C_2=conv(N''_B\cup V_B\cup \{v_2\})=conv(S_2\cup V_B\cup \{v_2\})$ in times
$O(|S_1|\log |S_1|)$ and $O(|S_2|\log |S_2|)$, respectively.
Finally, we compute $\pess(N_B)$ as the intersection $C_1\cap C_2$.
If $|C_1|$ and $|C_2|$ denote the numbers of vertices of $C_1$ and $C_2$, respectively,
a representation of $C_1\cap C_2$ can be computed in $O(\log |C_1|\cdot \log|C_2|)$ time
in the worst case:
If the point $\ell_1$ is not above the point $\ell_2$, then the boundary $C_1\cap C_2$ 
consists of the clockwise boundary path $L$ of $C_2$ connecting $b_1$ with $\ell_1$,
the segment connecting $\ell_1$ with $\ell_2$, the clockwise boundary path $U$ of $C_1$ 
connecting $\ell_2$ with $t_1$, and the similar path of $C_1$ that connects $t_1$ with $b_1$.
Hence, in this case, a representation of $C_1\cap C_2$ can be found in $O(1)$ time.
Otherwise, if $\ell_1$ is above $\ell_2$, a representation of $C_1\cap C_2$ is given
by such above paths and the intersection point $p$ between $L$ and $U$. Observe that both $L$ and $U$
are $x$-monotone paths, and in each of them the edges are sorted from left to right. To find $p$, we need
to find the edge of $L$ that intersects $U$. This can be done with a binary search over
the edges of $L$. Given any edge $e$ of $L$, deciding whether $e$ is to the left of $U$,
intersects $U$, or is to the right of $U$, can be done by querying to which side of $U$ is each
endpoint of $e$ (i.e.\ if the endpoint is or not inside $C_1$). Determining whether a given point belongs to 
a convex hull can be done with a binary search on the edges, running in
$O(\log k)$ time, where $k$ is the number of vertices. Then, each query costs
$O(\log |C_1|)$ time, and $O(\log |C_2|)$ queries are performed,
with a total running time of $O(\log |C_1|\cdot \log |C_2|)$.

\subsection{Deciding convex hulls intersection}\label{sec:hull-inter}

We show how to decide whether
$=\opt(N_R)$ and $\opt(N_B)$ are disjoint. A similar procedure
can be applied for $\pess(N_R)$ and $\pess(N_B)$.

Let $C_1=\opt(N_R)$ and $C_2=\opt(N_B)$.
Suppose that $\mbr(R)$ and $\mbr(B)$ have a corner intersection.
By Lemma~\ref{lem:VR-VB-inter}, we can orient clockwise the edges of $C_2$ and consider 
the sequence $S$ of consecutive edges that starts with the edge with source endpoint the 
bottom-left vertex of $\mbr(B)$, and ends with the edge with target endpoint the
top-right vertex of $\mbr(B)$. The sequence $S$ consists of three intervals: consecutive 
edges outside $C_1$ which point to $C_1$, consecutive edges which have at least one
endpoint inside $C_1$, and consecutive edges outside $C_1$ which do not point to $C_1$.
If the first edge of $S$ does not point to $C_1$, then $C_1$ and $C_2$ are disjoint.
In general, given any edge of $C_2$, querying to which interval of $S$ the edge belongs to
can be done with a binary search on the edges of $C_1$ from the top-right vertex to the
bottom-left vertex clockwise, in $O(\log|C_1|)$ time.
Furthermore, deciding whether there exists an edge of $S$ that has at least one endpoint
inside $C_1$ can be done with a binary search in $S$. The search performs $O(\log |C_2|)$
queries, and each query will cost $O(\log|C_1|)$ time. Deciding whether $C_1$ and $C_2$
are disjoint can thus be done in $O(\log|C_1|\cdot \log|C_2|)$ time.
To find a line separating $C_1$ and $C_2$, in the case where they are disjoint, we need to find the first
edge $e$ in $S$ that does not point to $C_1$. This edge can be found, similarly as above, in 
$O(\log|C_1|\cdot \log|C_2|)$ time. Let $\overline{e}$ denote the same edge $e$, but oriented
in the contrary direction. If $\overline{e}$ does not point to $C_1$, then the straight line
through $e$ is a separating line. Otherwise, in $O(\log|C_1|)$ time we can find the edge $e'$ of
$C_1$ pointed by $\overline{e}$, and the line through $e'$ is a separating line.

When $\mbr(R)$ and $\mbr(B)$ have a side intersection, a similar sequence $S$ can be considered.
In this case, $S$ is the sequence of consecutive edges that starts with the edge with source endpoint the 
bottom-right vertex of $\mbr(B)$, and ends with the edge with target endpoint the
top-right vertex of $\mbr(B)$. Deciding whether $C_1$ and $C_2$
are disjoint can be done in $O(\log|C_1|\cdot \log|C_2|)$ time.

\subsection{Filtering rectangles}\label{sec:filtering}

We show how to filter the rectangles of $N_R$ and $N_B$,
that is, to refine these sets by removing some elements, so that the
new $N_R$ and $N_B$ still satisfy the properties (1) and (2) of Definition~\ref{def:NR-NB}.
We use the natural way of removing rectangles, say from $N_B$, which consists in removing 
all rectangles completely contained
in the pessimistic convex hull (see the rectangle
in dashed lines in Figure~\ref{fig:NB6} and Figure~\ref{fig:pess-side}). 
If the rectangle has a part outside the pessimistic
convex hull, then it cannot be removed because such a part could contain blue points 
that are vertices of $conv(B\cup V_B)$. 

Note that a rectangle $N$ of $N_B$ is contained in 
$\pess(N_B)$ if and only if the two left vertices of $N$
belong to $\pess(N_B)$.
Thus, once we have computed $\pess(N_B)$, determining whether $N$ is contained
in $\pess(N_B)$ can be done by querying twice whether a point belongs to $\pess(N_B)$. In this
case, the two points are the two left vertices of $N$. Each query
runs in $O(\log k)$ time, where $k$ is the number of vertices of $\pess(N_B)$.
The running time to filter the rectangles of $N_B$ is 
then $O(|N_B|\cdot \log k)=O(|N_B|\cdot \log |N_B|)$.
Symmetrically,
a rectangle $N$ of $N_R$ is contained in 
$\pess(N_R)$ if and only if the two right vertices of $N$
belong to $\pess(N_R)$, and similar decision and filtering algorithms can be 
used, where the filtering algorithm runs in $O(|N_R|\cdot \log |N_R|)$ time.

\subsection{Separability algorithm}\label{sec:separab-algo}

The algorithm consists of an outer procedure and an inner procedure. 
The outer procedure (see the pseudocode {\tt DecideSeparability} of Figure~\ref{fig:algorithm0})
receives $R$ and $B$ as input, both represented in R-trees,
and decide the linear separability of $R$ and $B$. It also returns a separating line
in the positive case. In this procedure, we first initialize the rectangle set
$N_R$ as the rectangles contained in the root node of the R-tree of $R$, and the 
rectangle set $N_B$ as the rectangles contained in the root node of the R-tree of $B$.
This allows to compute both $\mbr(R)=\mbr(N_R)$ and $\mbr(B)=\mbr(N_B)$. 
Then, we proceed as follows:
If the intersection between $\mbr(R)$ and $\mbr(B)$ is empty, then we return a `yes'
answer together with an axis-aligned line containing an edge of $\mbr(R)$ that separates $R$ and $B$.
If $\mbr(R)$ and $\mbr(B)$ have a piercing intersection, then we return a 
`no' answer.
Otherwise, if $\mbr(R)$ and $\mbr(B)$ have a containment, corner, or side intersection,
then the inner procedure (see the pseudocode {\tt DecideSeparabilityCS} of Figure~\ref{fig:algorithm})
is called accordingly. This procedure decides the linear separability when 
$\mbr(R)$ and $\mbr(B)$ have corner or side intersection. Recall that when 
$\mbr(R)$ and $\mbr(B)$ have a containment intersection, the linear separation question
can be reduced to solve (at most) four instances of the same question in which 
$\mbr(R)$ and $\mbr(B)$ have a corner intersection (see Section~\ref{sec:preliminaries}).
This is done by extending the inner rectangle to contain one vertex of the outer rectangle.
The inner procedure is as follows:

We start by computing the vertex sets $V_R$ and $V_B$, according to 
the relative positions of $\mbr(R)$ and $\mbr(B)$, which are necessary to the algorithm 
to distinguish between a corner and a side intersection. Then, the following actions with
$N_R$ and $N_B$ are performed. We compute both $\opt(N_R)$ and $\opt(N_B)$ 
(see Section~\ref{sec:hulls-computation}), and test
whether $\opt(N_R)$ and $\opt(N_B)$ are disjoint (see Section~\ref{sec:hull-inter}).
If they are disjoint, then we report a `yes' answer and find a separating line 
(see Section~\ref{sec:hull-inter}). Otherwise, if $\opt(N_R)$ and $\opt(N_B)$ are
not disjoint, we continue as follows.
We compute both $\pess(N_R)$ and $\pess(N_B)$, and decide whether they are disjoint
(see sections~\ref{sec:hulls-computation} and~\ref{sec:hull-inter}).
If the are not disjoint, then we report a `no' answer. Otherwise, 
for each point set $X\in\{R,B\}$ such that $N_X$ is made of rectangles,
we filter $N_X$ (see Section~\ref{sec:filtering}),
and replace each remaining rectangle in $N_X$ by its child rectangles, or points, in
the corresponding R-tree. Observe that the new rectangles of $N_X$ are all a level down to the level of the former
rectangles in $N_X$. If at least one of the new $N_R$ and $N_B$ is made of
rectangles (it can happen that one of $N_R$ and $N_B$ is made of rectangles and the other one 
is made of points since the R-trees of $R$ and $B$ can have different heights),
then we repeat all this actions with these new $N_R$ and $N_B$. Otherwise,
if both $N_R$ and $N_B$ are made of points, we test whether $\opt(N_R)=conv(R)$ and
$\opt(N_B)=conv(B)$ are disjoint, and find a separating line in the positive case,
to finally decide whether $R$ and $B$ are linearly separable.

\begin{figure}
	\small
	\begin{framed}
	\begin{algorithm}{}{$\mathtt{DecideSeparability}(R,B)$:}
		$N_R$ $\leftarrow$ the set of rectangles in the root node of the R-tree of $R$\\
		$N_B$ $\leftarrow$ the set of rectangles in the root node of the R-tree of $B$\\
		$\mbr(R)\leftarrow \mbr(N_R)$\\
		$\mbr(B)\leftarrow \mbr(N_B)$\\
		\qif $\mbr(R)\cap \mbr(B)$ is empty \qthen\\
			\qreturn \qtrue
		\qelseif $\mbr(R)$ and $\mbr(B)$ have a piercing intersection \qthen\\
			\qreturn \qfalse
		\qelseif $\mbr(R)$ and $\mbr(B)$ have a containment intersection \qthen\\
			\qif $\mbr(B)\subset \mbr(R)$ \qthen\\
				\qfor {\bf each} vertex $v$ of $\mbr(R)$ \qdo\\
					\qcom{$\mbr(R)$ and $\mbr(\mbr(B)\cup \{v\})$ have a corner intersection}\\
					\qif $\mathtt{DecideSeparabilityCS}(N_R,N_B,\mbr(R),\mbr(\mbr(B)\cup \{v\}))$ \qthen\\
						\qreturn \qtrue
					\qfi
				\qrof\\
				\qreturn \qfalse
			\qelse\\
				\qfor {\bf each} vertex $v$ of $\mbr(B)$ \qdo\\
					\qcom{$\mbr(\mbr(R)\cup \{v\})$ and $\mbr(B)$ have a corner intersection}\\
					\qif $\mathtt{DecideSeparabilityCS}(N_R,N_B,\mbr(\mbr(R)\cup \{v\}),\mbr(B))$ \qthen\\
						\qreturn \qtrue
					\qfi
				\qrof\\
				\qreturn \qfalse
			\qfi
		\qelse\\
			\qcom{$\mbr(R)$ and $\mbr(B)$ have a corner or side intersection}\\
			\qreturn $\mathtt{DecideSeparabilityCS}(N_R,N_B,\mbr(R),\mbr(B))$
		\qfi
	\end{algorithm}
	\end{framed}
	\caption{\small{
		Algorithm to compute the linear separability of $R$ and $B$.
	}}
	\label{fig:algorithm0}
\end{figure}

\begin{figure}
	\small
	\begin{framed}
	\begin{algorithm}{}{$\mathtt{DecideSeparabilityCS}(N_R,N_B,\mbr(R),\mbr(B))$:}
		Compute $V_R$ and $V_B$ according to $\mbr(R)$ and $\mbr(B)$\\
		$\mathtt{opt}_R  \leftarrow \mathsf{OptimisticConvexHull}(N_R,V_R)$\\
		$\mathtt{opt}_B  \leftarrow \mathsf{OptimisticConvexHull}(N_B,V_B)$\\
		\qwhile (\qtrue) \qdo\\
			\qif $\mathtt{opt}_R$ and $\mathtt{opt}_B$ are disjoint \qthen\\
				\qreturn \qtrue
			\qelseif ($N_R$ is made of points) \qand ($N_B$ is made of points) \qthen\\
				\qreturn \qfalse
			\qelse \\
				\qif $N_R$ is made of rectangles \qthen\\
					$\mathtt{pess}_R  \leftarrow \mathsf{PessimisticConvexHull}(N_R,V_R)$
				\qelse\\
					$\mathtt{pess}_R \leftarrow \mathtt{opt}_R$
				\qfi\\
				\qif $N_R$ is made of rectangles \qthen\\
					$\mathtt{pess}_B  \leftarrow \mathsf{PessimisticConvexHull}(N_B,V_B)$
				\qelse\\
					$\mathtt{pess}_B \leftarrow \mathtt{opt}_B$
				\qfi\\
				\qif $\mathtt{pess}_R$ and $\mathtt{pess}_B$ are not disjoint \qthen\\
					\qreturn \qfalse
				\qelse\\
					\qif $N_R$ is made of rectangles \qthen\\
						$N_R \leftarrow \mathsf{Filter}(N_R)$\\
						$N_R \leftarrow \bigcup_{N\in N_R}\mathsf{children}(N)$\\
						$\mathtt{opt}_R  \leftarrow \mathsf{OptimisticConvexHull}(N_R,V_R)$
					\qfi\\
					\qif $N_B$ is made of rectangles \qthen\\
						$N_B \leftarrow \mathsf{Filter}(N_B)$\\
						$N_B \leftarrow \bigcup_{N\in N_B}\mathsf{children}(N)$\\
						$\mathtt{opt}_B  \leftarrow \mathsf{OptimisticConvexHull}(N_B,V_B)$
					\qfi
				\qfi
			\qfi
		\qend
	\end{algorithm}
	\end{framed}
	\caption{\small{
		Algorithm to compute the linear separability of $R$ and $B$ when
		$\mbr(R)$ and $\mbr(B)$ have a corner or side intersection:
	  	$\mathsf{OptimisticConvexHull}(\cdot)$ returns the optimistic convex 
		hull; $\mathsf{PessimisticConvexHull}(\cdot)$ returns the pessimistic
		convex hull; $\mathsf{Filter}(\cdot)$ removes from the argument
		the rectangles contained in the pessimistic convex hull; and
		$\mathsf{children}(N)$ returns the rectangles (or points) that are
		children of the rectangle $N$ in the corresponding R-tree.
	}}
	\label{fig:algorithm}
\end{figure}

In the following, we analyse the asymptotic running time of the algorithm in the worst case.
Let $m=|R|$ and $n=|B|$, and $r=O(\log m)$ and $b=O(\log n)$ denote the heights of the R-trees of $R$ and $B$, respectively.
Assume $r\le b$ without loss of generality.  
For $i=0,1,\ldots,r$, let $m_i$ denote the number of rectangles in the $i$-th level of
the R-tree of $R$, and $N^{(i)}_R$ the set of rectangles $N_R$ when $N_R$ is formed
by rectangles of the $i$-th level, where $|N^{(i)}_R|\le m_i$ and $m_{r}=m$.
Similarly, for $j=0,1,\ldots,b$, let $n_j$ denote the number of rectangles in 
the $j$-th level of the R-tree of $B$, and $N^{(j)}_B$ the set of rectangles $N_B$ when $N_B$ is formed
by rectangles of the $j$-th level, where $|N^{(j)}_B|\le n_j$ and $n_{b}=n$.
For level number $i=0,1,\ldots,r$, the algorithm in the worst case: 
\begin{itemize}
	\item computes both $\opt(N^{(i)}_R)$ and $\pess(N^{(i)}_R)$ in $O(|N^{(i)}_R|\cdot \log|N^{(i)}_R|)$ time;

	\item computes both $\opt(N^{(i)}_B)$ and $\pess(N^{(i)}_B)$ in $O(|N^{(i)}_B|\cdot \log|N^{(i)}_B|)$ time;

	\item decides $\opt(N^{(i)}_R)\cap \opt(N^{(i)}_B)=\emptyset$ 
	in $O(\log |N^{(i)}_R| \cdot \log|N^{(i)}_B|)$ time;

	\item decides $\pess(N^{(i)}_R)\cap \pess(N^{(i)}_B)=\emptyset$ 
	in $O(\log |N^{(i)}_R| \cdot \log|N^{(i)}_B|)$ time;

	\item filters $N^{(i)}_R$ in $O(|N^{(i)}_R|\cdot \log|N^{(i)}_R|)$ time, for $i<r$; and

	\item filters $N^{(i)}_B$ in $O(|N^{(i)}_B|\cdot \log|N^{(i)}_B|)$ time, for $i<b$.
\end{itemize}
For level number $j=r+1,\ldots,b$, the algorithm in the worst case: 
\begin{itemize}
	\item computes both $\opt(N^{(j)}_B)$ and $\pess(N^{(j)}_B)$ in $O(|N^{(j)}_B|\cdot \log|N^{(j)}_B|)$ time;

	\item decides $\opt(N^{(r)}_R)\cap \opt(N^{(j)}_B)=\emptyset$ 
	in $O(\log |N^{(r)}_R| \cdot \log|N^{(j)}_B|)$ time;

	\item decides $\pess(N^{(r)}_R)\cap \pess(N^{(j)}_B)=\emptyset$ 
	in $O(\log |N^{(r)}_R| \cdot \log|N^{(j)}_B|)$ time; and

	\item filters $N^{(j)}_B$ in $O(|N^{(j)}_B|\cdot \log|N^{(j)}_B|)$ time, for $j<b$.
\end{itemize}
Summing up, the running time in the worst case is:
\begin{eqnarray*}
	&   & \sum_{i=0}^{r}\left(O\left(|N^{(i)}_R|\cdot \log|N^{(i)}_R|\right) + 
				O\left(|N^{(i)}_B|\cdot \log|N^{(i)}_B|\right)  +
				O\left(\log |N^{(i)}_R|\cdot \log|N^{(i)}_B|\right)\right) + \\
	&   & \sum_{i=r+1}^{b}\left(O\left(|N^{(i)}_B|\cdot \log|N^{(i)}_B|\right) +
				O\left(\log|N^{(r)}_R|\cdot \log|N^{(i)}_B|\right)\right)\\
	& = &  \sum_{i=0}^{r}\Bigl(O(m_i\log m_i) +O(n_i\log n_i) + O(\log m_i\log n_i)\Bigr) + 
	       \sum_{i=r+1}^{b}\Bigl(O(n_i \log n_i)+O(\log m \log n_i)\Bigr) \\
	& = &  \sum_{i=0}^{r}\Bigl(O(m_i\log m_i) +O(n_i\log n_i)\Bigr) + 
	       \sum_{i=r+1}^{b}\Bigl(O(n_i \log n_i)+O(\log m \log n_i)\Bigr)\\
	& = &  O\left(\sum_{i=0}^{r} m_i\log m_i + \sum_{i=0}^{b} n_i\log n_i + 
	       \sum_{i=r+1}^{b} \log m \log n_i\right) \\
	& = &  O\left(\log m \cdot \sum_{i=0}^{r} m_i + \log n \cdot \sum_{i=0}^{b} n_i + (b-r)\log m \log n\right) \\
	& = &  O\left(\log m \cdot \sum_{i=0}^{r} m_i + \log n \cdot \sum_{i=0}^{b} n_i + \log m \log^2 n\right).
\end{eqnarray*}
Since an R-tree has the property that every node contains at least two children nodes,
we have that $m_{r}=m$, $m_{r-1}\le m/2$, $m_{r-2}\le m/4$, $m_{r-3}\le m/8$, and so on.
That is, $m_i\le m/2^{r-i}$ for $i=0,1,\ldots,r$.
Similarly, $n_j\le n/2^{b-j}$ for $j=0,1,\ldots,b$. The above running time is then:
\begin{eqnarray*}
   O\left(\log m \cdot \sum_{i=0}^{r}m/2^i+ \log n \cdot \sum_{i=0}^{b}n/2^i + \log m \log^2 n\right)
   & = & O\left(m \log m +  n\log n + \log m \log^2 n\right) \\
   & = & O(m\log m + n\log n).
\end{eqnarray*}

The worst case of our algorithm occurs when all nodes of the two R-trees, and all rectangles and points,
need to be loaded to decide the linear separability of $R$ and $B$. This happens in the following
example. Suppose that all elements of $B$ belong to the line $y=x$, for example, $B=\{(i,i):i=1,2,\ldots,n\}$,
and that $|R|=|B|=n$ with $R=\{(i-\varepsilon,i+\varepsilon):i=1,2,\ldots,n\}$ for $\varepsilon=1/2$. 
In this case, $R$ and $B$ are linearly separable, and 
$\mbr(R)$ and $\mbr(B)$ have a corner intersection, where $\mbr(R)$ contains the top-left vertex of
$\mbr(B)$, and $\mbr(B)$ contains the bottom-right vertex of $\mbr(R)$. Consider any
step of our algorithm (refer to Figure~\ref{fig:algorithm}), with $N_R$ and $N_B$
representing the rectangles from the R-trees of $R$ and $B$, respectively. 
Observe that in every rectangle of $N_B$ the diagonal connecting the bottom-left vertex
with the top-right vertex is contained in the line $y=x$. This implies that the pessimistic convex hull 
$\pess(N_B)$ equals the triangle with vertices $(1,1)$, $(n,1)$, and $(n,n)$, and that no rectangle of $N_B$ is contained
in $\pess(N_B)$. A similar situation occurs with $N_R$: the pessimistic convex hull $\pess(N_R)$ equals
the triangle with vertices $(1-\varepsilon,1+\varepsilon)$, $(1-\varepsilon,n+\varepsilon)$, 
and $(n-\varepsilon,n+\varepsilon)$, and no rectangle of $N_R$ is contained
in $\pess(N_R)$. Then, no rectangle of $N_R$ or $N_B$ can be discarded in any step of the algorithm.
Furthermore, $\pess(N_R)$ and $\pess(N_B)$ are disjoint, whereas 
the optimistic convex hulls $\opt(N_R)$ and $\opt(N_B)$ are not disjoint (due to the way we choose $\varepsilon$).
All of these observations imply that the algorithm will stop alfter loading all points, hence all 
nodes and rectangles, from both R-trees.

\subsection{The convex hull of a point set}\label{sec:mono-hull-compute}

Let $P$ be a finite point set in the plane, given in an R-tree. In this section,
we present an algorithm to compute $conv(P)$. Let $N_P$ be a set of rectangles
from the R-tree of $P$ such that $N_P$ satisfies $conv(P)\subseteq conv(N_P)$
(similar as Definition~\ref{def:NR-NB}). To compute $conv(P)$, we first
discard rectangles $N$ from $N_P$ such that all points of $P$ contained
in $N$ are not vertices of $conv(P)$. After that, we replace each rectangle $N$
that remains in $N_P$ by its child rectangles (or points) in the R-tree. We stop
when all elements of $N_P$ are points, and return $conv(N_P)$. 
Let $v_1$, $v_2$, $v_3$, and $v_4$ be the top-left, bottom-left, bottom-right, and top-right
vertices of $\mbr(P)$, respectively.
To discard rectangles
from $N_P$, we compute the following convex hulls by following ideas similar to those
given in Section~\ref{sec:hulls-computation}: 
$C_1=conv(\{\SE(N)\mid N\in N_P\} \cup \{v_2,v_3,v_4\})$,
$C_2=conv(\{\NE(N)\mid N\in N_P\} \cup \{v_1,v_3,v_4\})$,
$C_3=conv(\{\NW(N)\mid N\in N_P\} \cup \{v_1,v_2,v_4\})$, and
$C_4=conv(\{\SW(N)\mid N\in N_P\} \cup \{v_1,v_2,v_3\})$.
Given a rectangle $N$ of $N_P$, if $N$ is contained in the intersection
$C_1\cap C_2\cap C_3\cap C_4$, then no point of $P$ contained in $N$ 
can be a vertex of $conv(P)$. Furthermore, $N$ is contained in 
$C_1\cap C_2\cap C_3\cap C_4$ if and only if the top-left,
bottom-left, bottom-right, and top-right vertices of $N$ are contained
in $C_1$, $C_2$, $C_3$, and $C_4$, respectively. Once such four convex hulls are computed, 
these four decisions can be made in times $O(\log |C_1|)$, $O(\log |C_2|)$,
$O(\log |C_3|)$, and $O(\log |C_4|)$, respectively.
The algorithm to compute $conv(P)$ is described in the
pseudocode of Figure~\ref{fig:convexhull-alg}. The running time
is $O(n\log n)$, where $n$ is the number of points, and can
be obtained from arguments similar to those of Section~\ref{sec:separab-algo}.

\begin{figure}
	\small
	\begin{framed}
	\begin{algorithm}{}{$\mathtt{ConvexHull}(P)$:}
		$N_P$ $\leftarrow$ the set of rectangles in the root node of the R-tree of $P$\\
		$\mbr(P)\leftarrow \mbr(N_P)$\\
		\qrepeat\\
			$C_1 \leftarrow conv(\{\SE(N)\mid N\in N_P\} \cup \{v_2,v_3,v_4\})$\\
			$C_2 \leftarrow conv(\{\NE(N)\mid N\in N_P\} \cup \{v_1,v_3,v_4\})$\\
			$C_3 \leftarrow conv(\{\NW(N)\mid N\in N_P\} \cup \{v_1,v_2,v_4\})$\\
			$C_4 \leftarrow conv(\{\SW(N)\mid N\in N_P\} \cup \{v_1,v_2,v_3\})$\\
			$N_P \leftarrow N_P \setminus \{N\in N_P\mid N\subset C_1\cap C_2\cap C_3\cap C_4\}$\\
			$N_P \leftarrow \bigcup_{N\in N_P}\mathsf{children}(N)$
		\qendrepeat\\
		\quntil $N_P$ is made of points\\
		\qreturn $conv(N_P)$
	\end{algorithm}
	\end{framed}
	\caption{\small{
		Algorithm to compute the convex hull of a point set $P$
		given in an R-tree.
	}}
	\label{fig:convexhull-alg}
\end{figure}

\section{Experimental results}\label{sec:experiments}

In this section, we describe the experiments that we implemented to evaluate the performance of our separability testing algorithm in terms of running time (via counting the number of access to nodes of the R-trees) and memory usage. The algorithm was implemented in the \texttt{C++} language, using the implementation of the R-tree data structure of the library \texttt{LibSpatialIndex}~\cite{LibSpatialIndex}. Nodes of size 1K were used to build the R-trees. The experiments were executed in a {\tt  Lenovo ThinkPad x240} computer, with 8GB of {\tt RAM} memory, and an {\tt Intel Core i5 4300U} microprocessor, and both real and synthetic data were considered. 

\subsection{Real data}

We consider a first data set consisting of 200 thousands of MBRs representing spatial objects of California roads, and a second data set consisting of 2.2 millions of MBRs representing spatial objects of rivers of Iowa, Kansas, Missouri, and Nebraska~\cite{realdata}. From the first set we generate a set of red points by taking from each MBR its center point, and perform a similar operation to the second data set to obtain a set of blue points. To test our algorithm, we mapped both point sets to the space $[0,1]\times[0,1]$. Each colored point set was stored in a different R-tree, and the numbers of blocks used in the R-trees are shown in Table~\ref{tabla_resultados_reales}. In Figure~\ref{fig:reales-a}, we make a graphic representation of both colored point sets, where we plot only about the 10\% of the points of each set. In Figure~\ref{fig:reales-b}, we draw the MBR of each colored point set, showing that they have a side intersection with a considerably high overlapping, precisely, the common area of the MBRs is above the 98\% of the total area.
In Table~\ref{tabla_resultados_reales} we show the results of the execution of our linear separability testing algorithm: a 2.79\% of the nodes (i.e.\ memory blocks) of the red R-tree are accessed, whereas a 1.36\% of the nodes of the blue R-tree are accessed. To solve this particular instance of the problem, only 40Kb of main memory is required.

\begin{figure}[ht]
	\centering
	\subfloat[] {
		\includegraphics[scale = 0.42]{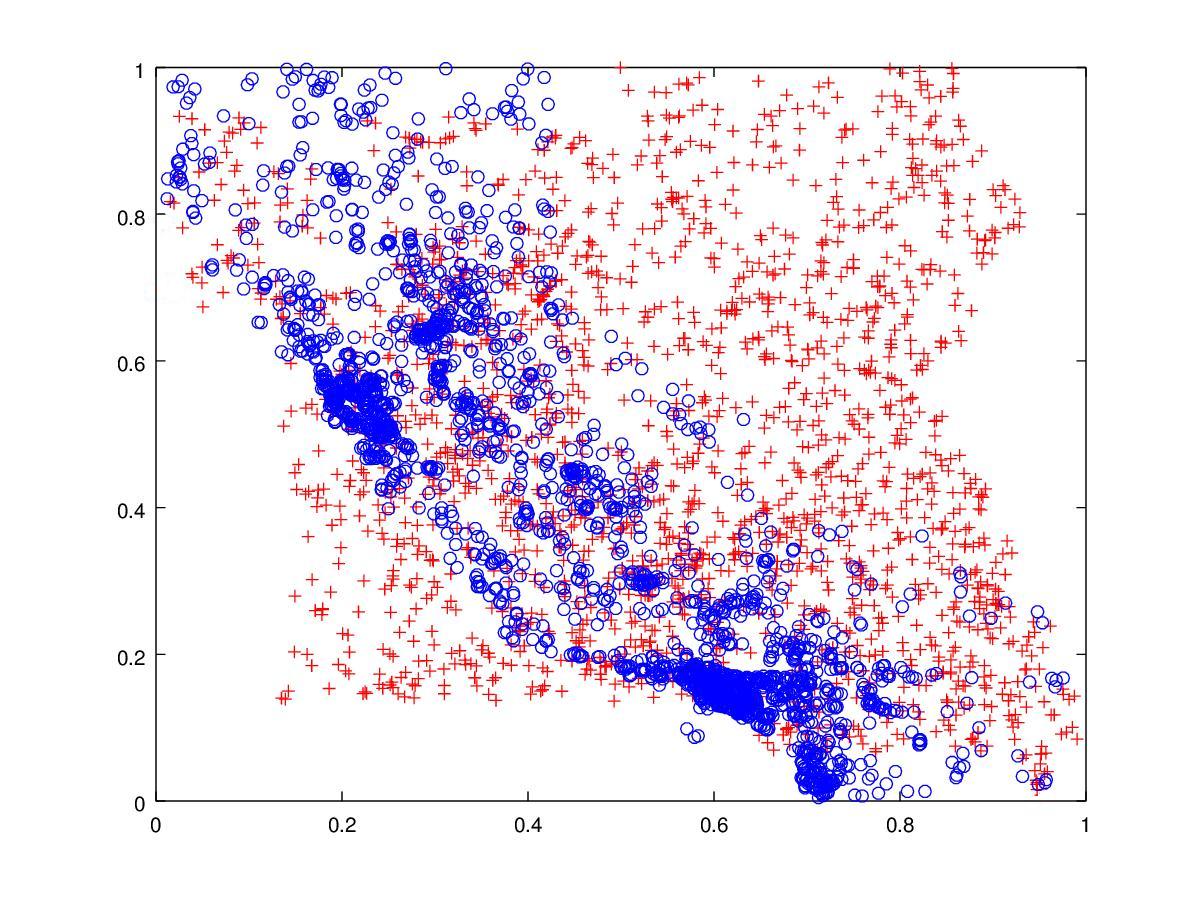}
        \label{fig:reales-a}
	}
	\subfloat[]{
		\includegraphics[scale = 0.42]{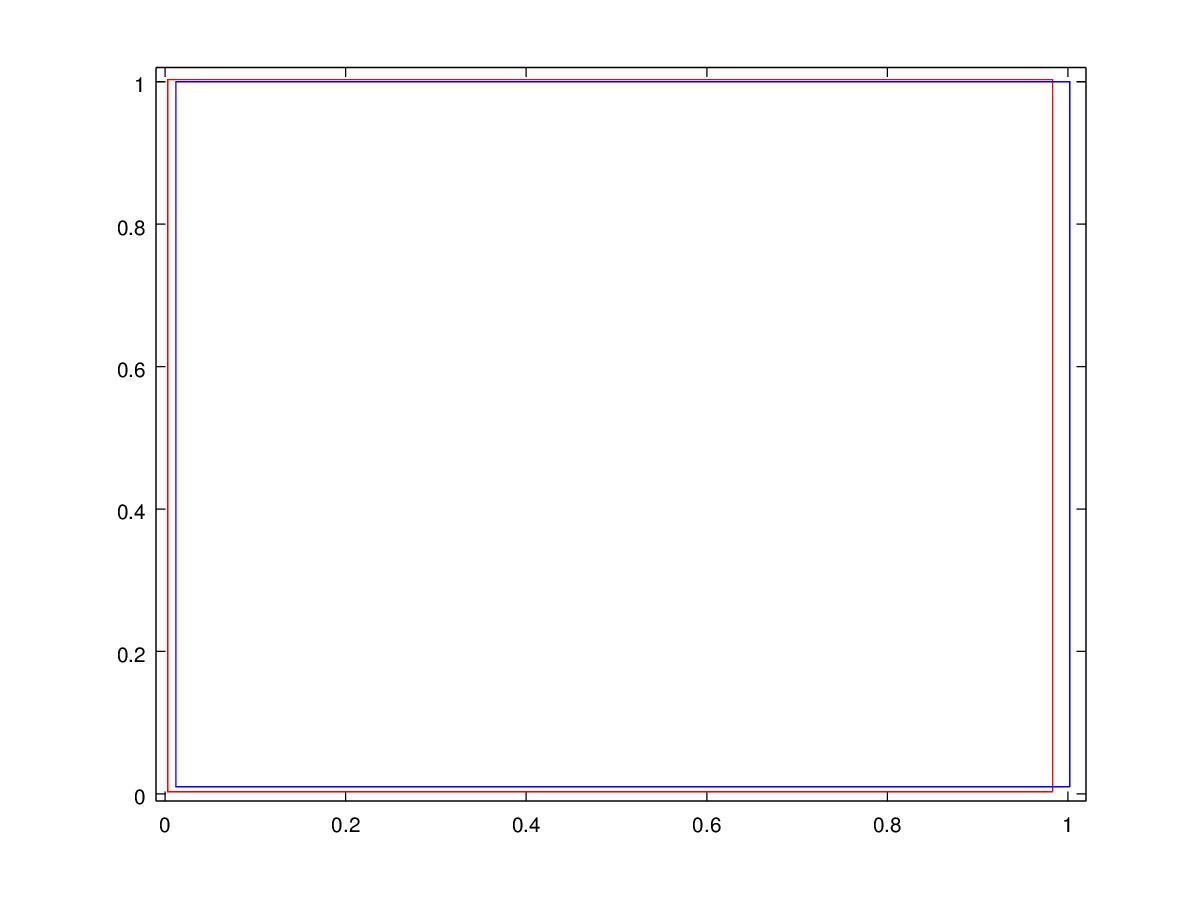}
        \label{fig:reales-b}
	}
	\caption{\small{
    	a) Colored point sets generated from real data. Each 
    	red points is represented by a $+$, and each blue point by a $\circ$.
    	b) The MBR of each point set, where the MBRs has a side intersection and their
    	intersection area is above the 98\% of the total area.
    }}  
	\label{fig_reales}
\end{figure}

\begin{table}[ht]
	\centering
	\begin{tabular}{|l|r|c|c|}
		\hline
		 & size     & \% nodes accessed & size of R-tree (\# of blocks)\\
		\hline
		Red points & 200,000            & 2.79\%          & 12,178                  \\
		\hline
		Blue points & 2,200,000          & 1.36\%          & 35,965                  \\
		\hline
	\end{tabular}
	\caption{\small{Results of experiments on real data.}}
	\label{tabla_resultados_reales}
\end{table}

\subsection{Synthetic data}

We extecute our algorithm on several synthetic data sets, each data set consisting of colored point sets randomly generated in the range $[0,1]\times[0,1]=[0,1]^2$. To generate a point set $R\cup B$, we proceed in the following steps:
\begin{enumerate}
\item We pick the number of points that each color class will contain. This number is either 1, 2, 5, or 10 millions.

\item We define two rectangles $R_r,R_b\subset [0,1]^2$ of equal areas, so that $R$ and $B$ will be generated inside $R_r$ and $R_b$, respectively, and the area of $R_r\cap R_b$ represents a given percent of the areas of $R_r$ and $R_b$. We select such a percent among 1\%, 5\%, 10\%, and 50\%. Furthermore, we also fix the type of intersection of $R_r$ and $R_b$ (hence the type of intersection of $\mbr(R)$ and $\mbr(B)$): corner or side. 

\item We select the distribution in which each point set ($R$ and $B$) is generated inside its corresponding rectangle. We consider two distributions: uniform and Gaussian.
\end{enumerate}
In total, we run our algorithm on 64 synthetic data sets, accounting from 4 possibilities for the generated number of points, times 4 percents of intersection area, times 2 types of intersection of $\mbr(R)$ and $\mbr(B)$, times 2 distributions. 
In Figure~\ref{fig_graficos_mbr_gauss}, we show examples of colored point sets generated with a Gaussian distribution having their MBRs a corner intersection.

For each possible number of points in the generated point sets, the average size of the R-trees, expressed as the number of disk blocks (i.e.\ nodes), is shown in Table~\ref{tabla_bloques_rtree}. We also measured the percentage of the nodes of the R-trees that were accessed by our algorithm, as shown in Figure~\ref{fig_graficos}, Table~\ref{tabla_access_nodes_uniform}, and Table~\ref{tabla_access_nodes_gauss}.

\begin{table}[]
	\centering
	\begin{tabular}{|r|r|r|r|}
		\hline
		Size of point sets & Red R-tree & Blue R-tree & Total    \\
		(in millions) & (disk blocks)  & (disk blocks) &    (disk blocks)  \\
		\hline
		\hline
		1      & 61,386         & 61,388          & 122,774   \\
		\hline
		2      & 122,210        & 122,226         & 244,436   \\
		\hline
		5      & 304,494        & 304,858         & 609,352    \\
		\hline
		10     & 608,671        & 609,749         & 1,218,420  \\
		\hline
	\end{tabular}
	\caption{\small{Number of blocks (i.e.\ nodes) used in average to build the R-trees.}}
	\label{tabla_bloques_rtree}
\end{table}

\begin{table}[]
	\centering
	\begin{tabular}{|c|c|c|c|c|c|c|c|c|}
		\hline
		\multirow{4}{*}{\begin{tabular}{r}Size of point sets \\ (in millions)\end{tabular}} & \multicolumn{8}{c|}{Type of intersection}                                   \\ \cline{2-9} 
                      & \multicolumn{4}{c|}{Corner}          & \multicolumn{4}{c|}{Side}            \\ \cline{2-9} 
                      & \multicolumn{4}{c|}{\% of intersection} & \multicolumn{4}{c|}{\% of intersection} \\ \cline{2-9} 
                      & 1\%      & 5\%     & 10\%   & 50\%   & 1\%     & 5\%     & 10\%    & 50\%   \\ \hline
		\hline
		1                     & 0.54     & 0.31    & 0.41   & 0.34   & 0.34    & 1.14    & 1.08    & 2.11   \\ \hline
		2                     & 0.23     & 0.13    & 0.09   & 0.35   & 0.22    & 0.39    & 0.58    & 0.19   \\ \hline
		5                     & 0.08     & 0.11    & 0.03   & 0.05   & 0.12    & 0.36    & 0.28    & 0.36   \\ \hline
		10                    & 0.18     & 0.03    & 0.03   & 0.44   & 0.09    & 0.18    & 0.2     & 0.27   \\ \hline
	\end{tabular}
	\caption{\small{Percentage of the nodes of the R-trees that were accessed, 
    in the cases where the points were generated by using the uniform distribution.}}
	\label{tabla_access_nodes_uniform}
\end{table}

\begin{table}[ht]
	\centering
	\begin{tabular}{|c|c|c|c|c|c|c|c|c|}
		\hline
		\multirow{4}{*}{\begin{tabular}{r}Size of point sets \\ (in millions)\end{tabular}} & \multicolumn{8}{c|}{Type of intersection}                                   \\ \cline{2-9} 
                      & \multicolumn{4}{c|}{Corner}            & \multicolumn{4}{c|}{Side}                \\ \cline{2-9} 
                      & \multicolumn{4}{c|}{\% of intersection} & \multicolumn{4}{c|}{\% of intersection} \\ \cline{2-9} 
                      & 1\%     & 5\%     & 10\%    & 50\%   & 1\%     & 5\%     & 10\%    & 50\%   \\ \hline
        \hline
		1                     & 0.45    & 0.13    & 0.41    & 0.13   & 0.41    & 0.16    & 0.29    & 0.01   \\ \hline
		2                     & 0.25    & 0.24    & 0.23    & 0.17   & 0.18    & 0.08    & 0.09    & 0.19   \\ \hline
		5                     & 0.03    & 0.03    & 0.03    & 0.03   & 0.08    & 0.00       & 0.01    & 0.00   \\ \hline
		10                    & 0.06    & 0.05    & 0.06    & 0.05   & 0.01    & 0.00       & 0.01    & 0.00      \\ \hline
	\end{tabular}
	\caption{\small{Percentage of the nodes of the R-trees that were accessed, 
    in the cases where the points were generated by using the Gaussian distribution.}}
	\label{tabla_access_nodes_gauss}
\end{table}

The results indicate that when points are generated by a uniform distribution, the number of R-tree nodes accessed increases as the percentage of intersection between $R_r$ and $R_b$ is increased (see Figure~\ref{fig:corner-uniforme} and Figure~\ref{fig:side-uniforme}).
For example, for 1 million of points generated and a 50\% of area of intersection, and a side intersection, it is needed to access to a $2.11\%$ of the nodes of the R-trees, whereas under the same conditions but a 1\% of intersection, it is needed to access a $0.34\%$ of the nodes.
For point sets generated with the Gaussian distribution (Figure~\ref{fig:corner-gauss} and Figure~\ref{fig:side-gauss}), the percentage of nodes accessed seems to not depend on the percentage of intersection. For example, for 1 million of points generated, if the \% of area of intersection between the MBRs is 1.0\% or 10.0\%, it is needed to access a 0.45\% or 0.41\% of the nodes, respectively.
In Figures~\ref{fig:corner-uniforme},~\ref{fig:side-uniforme},~\ref{fig:corner-gauss}, and~\ref{fig:side-gauss}, we can note that when the number of points generated is increased, the percentage of nodes accessed decreases, being in all cases less that 0.1\%.

In Table~\ref{tabla_memory_uniform} and Table~\ref{tabla_memory_gauss}, we note the amount of memory used by the algorithm for the point sets with uniform and Gaussian distributions. This includes the amount of memory used for the node lists and convex hulls (optimistic and pessimistic). We can also note that the required ranges between 15Kb and 47Kb. Furthermore, the memory is not affected when we increase the side of the point sets generated. In Figure~\ref{fig_graficos_mem}, it is shown that the amount of memory required has a similar behavior, independently of the type of intersection of the MBRs and the type of distribution. 

\begin{table}[]
	\centering
	\begin{tabular}{|c|c|c|c|c|c|c|c|c|}
		\hline
		\multirow{4}{*}{\begin{tabular}{r}Size of point sets \\ (in millions)\end{tabular}} & \multicolumn{8}{c|}{Type of intersection}                                   \\ \cline{2-9} 
                      & \multicolumn{4}{c|}{Corner}          & \multicolumn{4}{c|}{Side}            \\ \cline{2-9} 
                      & \multicolumn{4}{c|}{\% of intersection} & \multicolumn{4}{c|}{\% of intersection} \\ \cline{2-9} 
                      & 1\%      & 5\%     & 10\%   & 50\%   & 1\%     & 5\%     & 10\%    & 50\%   \\ \hline
		\hline
		1                     & 24     & 24    & 25   & 26   & 26    & 27    & 27    & 29   \\ \hline
		2                     & 42     & 41   & 46   & 46   & 44    & 43    & 43    & 43   \\ \hline
		5                     & 9    & 11    & 9   & 11   & 14 & 12   & 14   & 12   \\ \hline
		10                    & 17     & 16    & 17  & 21   & 18    & 19    & 18     & 17   \\ \hline
	\end{tabular}
	\caption{\small{Memory usage (in Kb) in sets with uniform distribution.}}
	\label{tabla_memory_uniform}
\end{table}

\begin{table}[ht]
	\centering
	\begin{tabular}{|c|c|c|c|c|c|c|c|c|}
		\hline
		\multirow{4}{*}{\begin{tabular}{r}Size of point sets \\ (in millions)\end{tabular}} & \multicolumn{8}{c|}{Type of intersection}                                   \\ \cline{2-9} 
                      & \multicolumn{4}{c|}{Corner}            & \multicolumn{4}{c|}{Side}                \\ \cline{2-9} 
                      & \multicolumn{4}{c|}{\% of intersection} & \multicolumn{4}{c|}{\% of intersection} \\ \cline{2-9} 
                      & 1\%     & 5\%     & 10\%    & 50\%   & 1\%     & 5\%     & 10\%    & 50\%   \\ \hline
        \hline
		1                     & 26    & 23    & 25    & 23   & 33    & 26    & 28    & 25   \\ \hline
		2                     & 44   & 44    & 44    & 43   & 47    & 45    & 43    & 43   \\ \hline
		5                     & 8    & 8    & 9   & 8   & 15 & 7 & 9    & 7      \\ \hline
		10                    & 19    & 17    & 17    & 17  & 19    & 20       & 17    & 14     \\ \hline
	\end{tabular}
	\caption{\small{Memory usage in sets with Gaussian distribution.}}
	\label{tabla_memory_gauss}
\end{table}

\begin{figure}[ht]
	\centering
	\subfloat[\small{Corner, 1\% of intersection.}]{
		\includegraphics[scale = 0.53]{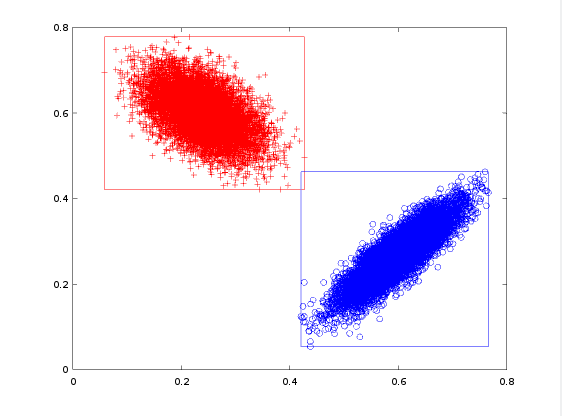}
	}
	\subfloat[\small{Corner, 5\% of intersection.}]{
		\includegraphics[scale = 0.53]{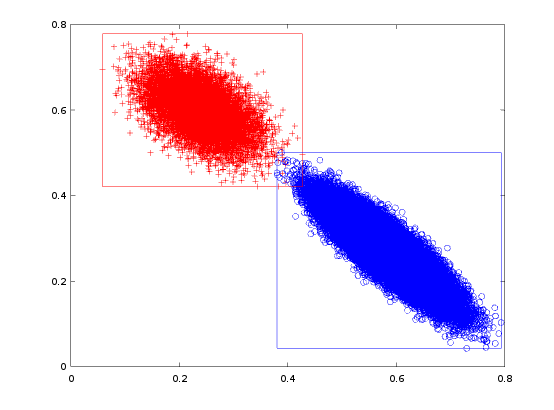}
	}
	\caption{\small{Synthetic data with a Gaussian distribution.}} 
	\label{fig_graficos_mbr_gauss}
\end{figure}

\begin{figure}[ht]
	\centering
	\subfloat[\small{Sets with uniform distribution and corner intersection.}]{
		\includegraphics[scale = 0.4]{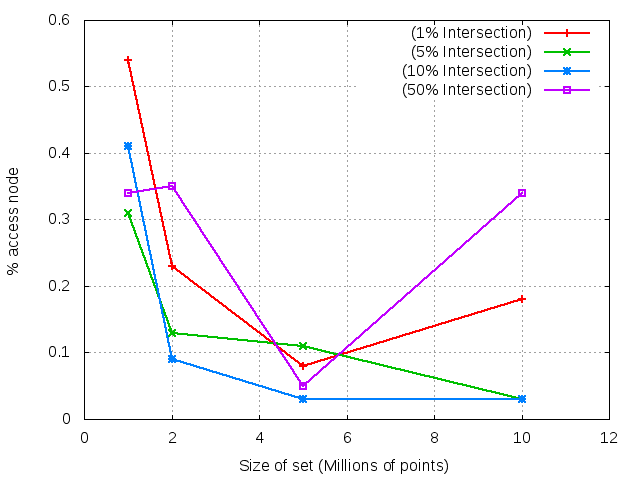}
        \label{fig:corner-uniforme}
	}\hspace{0.3cm}
	\subfloat[\small{Sets with uniform distribution and side intersection.}]{
		\includegraphics[scale = 0.4]{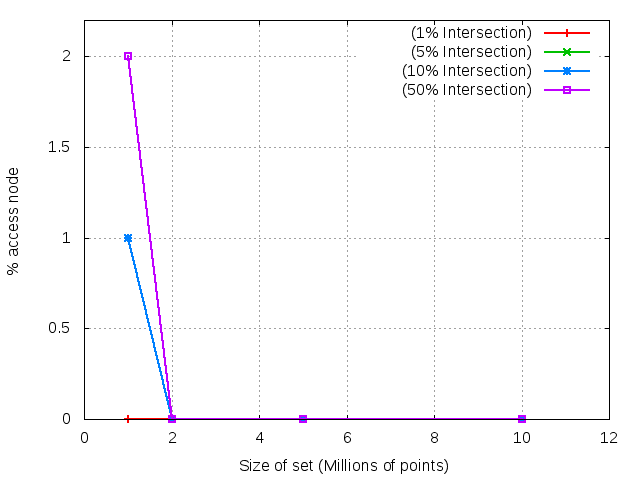}
        \label{fig:side-uniforme}
	}\\
	\subfloat[\small{Sets with Gaussian distribution and corner intersection.}]{
		\includegraphics[scale = 0.4]{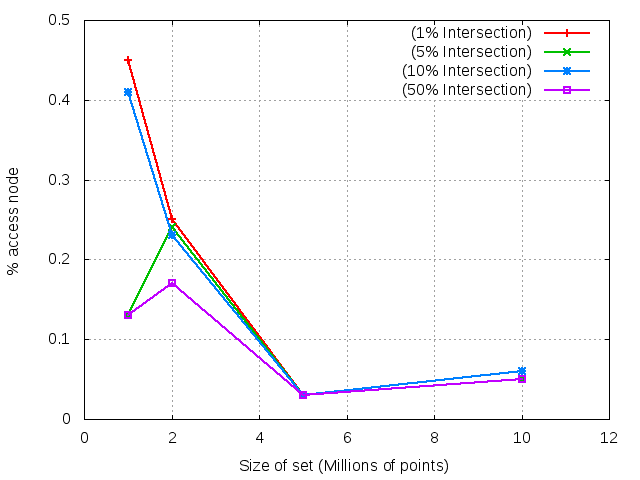}
        \label{fig:corner-gauss}
	}\hspace{0.3cm}
	\subfloat[\small{Sets with Gaussian distribution and side intersection.}]{
		\includegraphics[scale = 0.4]{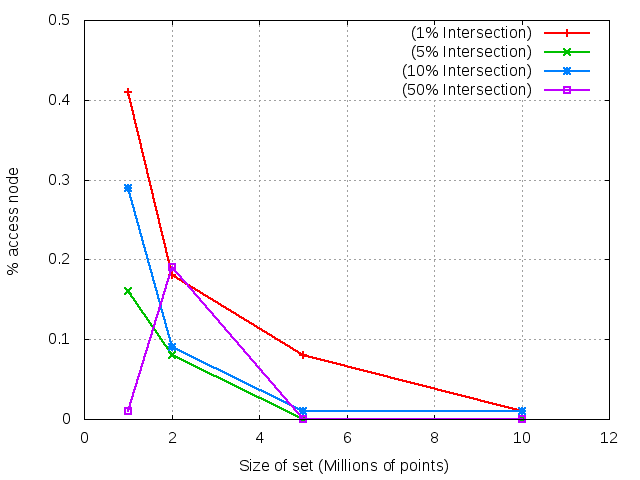}
        \label{fig:side-gauss}
	}
    %\\
	%\subfloat[Datos reales vs datos sintéticos con distribución gaussiana]{
	%	\includegraphics[scale = 0.32]{real-gauss}
	%}
	\caption{\small{Percentage of nodes accessed.}} 
	\label{fig_graficos}
\end{figure}

\begin{figure}[ht]
	\subfloat[\small{Uniform distribution, corner intersection.}]{
		\includegraphics[scale = 0.4]{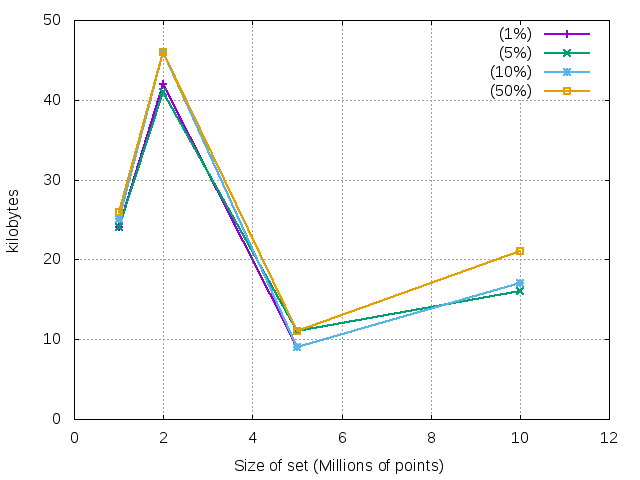}
	}\hspace{0.3cm}
	\subfloat[\small{Uniform distribution, side intersection.}]{
		\includegraphics[scale = 0.4]{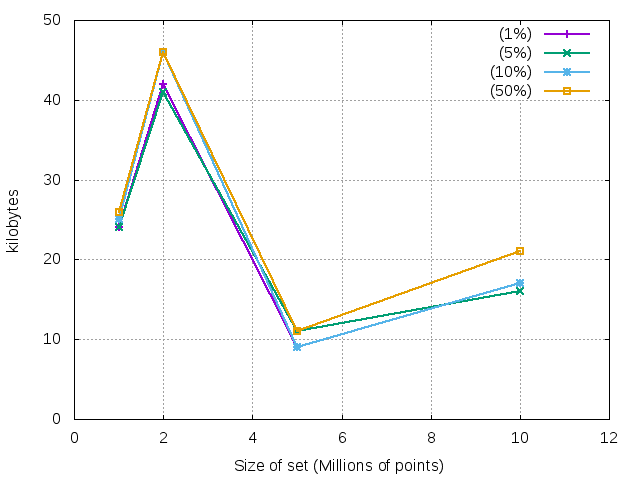}
	}\\
	\subfloat[\small{Gaussian distribution, corner intersection.}]{
		\includegraphics[scale = 0.4]{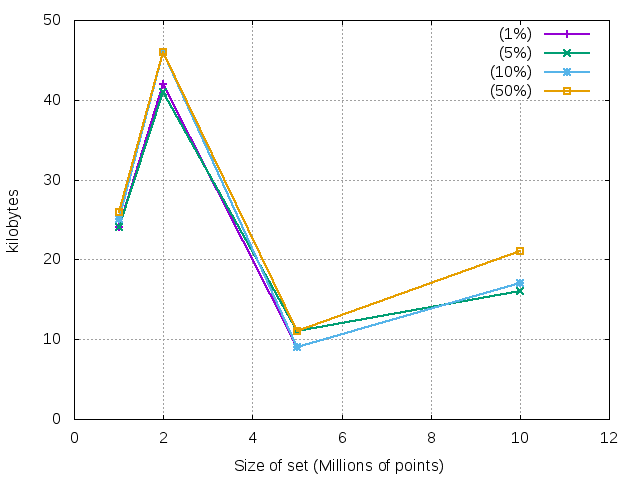}
	}\hspace{0.3cm}
	\subfloat[\small{Gaussian distribution, side intersection.}]{
		\includegraphics[scale = 0.4]{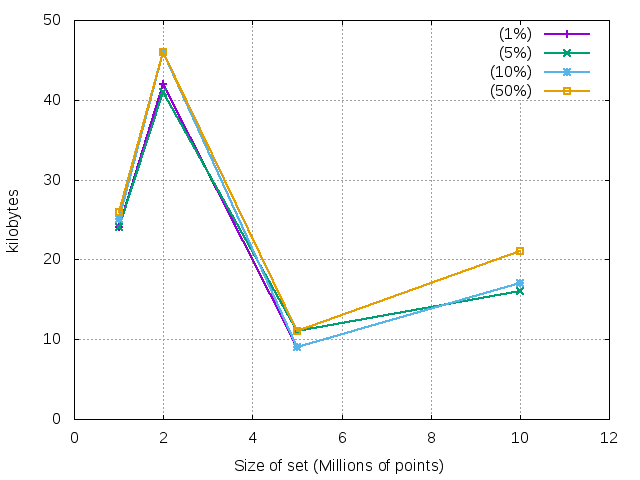}
	}
	\caption{\small{Memory usage.}} 
	\label{fig_graficos_mem}
\end{figure}

\section{Conclusions}\label{sec:conclusions}

In this paper, we have proposed an algorithm to decide the linear separability of two point sets of cardinalities $n$ and $m$, respectively, both sets stored in a different R-tree. The algorithm takes advantage of the properties of the R-trees in order to access as less nodes as possible. The running time complexity in the worst case is within $O(m \log m + n \log n)$. With the goal of evaluating the performance of the algorithm in practice, we designed several experiments with both real and synthetic point sets, and an implementation of the algorithm was run in each experiment. The results of the experiments showed that the algorithm performs few accesses to disk (i.e.\ accesses to nodes of the R-trees), uses a small amount of RAM memory and a low computation time.

Our algorithm expands the use of the R-trees, a multidimensional data structure well used in several spatial database systems such as Postgres and Oracle. According to the bibliography review, and to the best of our knowledge, this is the first algorithm that tackles the geometric separability of massive spatial object sets stored in secondary storage data structures. 

For future work, we propose the study of other types of geometric separability problems when the input is given in R-trees or other secondary-storage spatial data structures, for example separating red and blue points by axis parallel rectangles, wedges, or constrained polylines. We also propose to design an extension of this algorithm to work in dimensions higher that two.

\section*{Acknowledgements}

C.~T.\ was supported by CONICYT scholarship PCHA/MagisterNacional/2015-22151665 of the Government of Chile and the research group Bases de Datos 132019 GI/EF funded by Universidad del B\'io-B\'io (Chile).\\
P.~P-L.\ was supported by project Millennium Nucleus Information and Coordination in Networks ICM/FIC RC130003 (Chile).\\
G.~G.\ was supported by the research group Bases de Datos 132019 GI/EF, and the research project  Dise\~no e Implementaci\'on de Algoritmos Geom\'etricos en el contexto de Bases de Datos Espaciales 142719 3/R, both funded by Universidad del B\'io-B\'io (Chile).           

\small

\bibliographystyle{abbrv}
\bibliography{refs}

\end{document}